\documentclass[journal,twocolumn,10pt]{IEEEtran}
\usepackage{amsmath,epsfig}
\usepackage{amssymb}
\usepackage{booktabs}
\usepackage{amsthm}
\usepackage{graphicx}
\usepackage{caption}
\usepackage{subcaption}
\usepackage{rotating}
\usepackage{floatflt} 
\usepackage{paralist} 
\usepackage{algorithm} 
\usepackage{xcolor}
\usepackage{color}
\usepackage{epstopdf}
\usepackage[normalem]{ulem}
\usepackage{float}
\usepackage{todonotes}

\theoremstyle{definition}
\newtheorem{defn}{Definition}
\newcommand{\mypar}[1]{{\bf #1.}}
\newtheorem{myLem}{Lemma}

\newtheorem{myThm}{Theorem}
\newtheorem{myCorollary}{Corollary}

\newcommand{\R}{\ensuremath{\mathbb{R}}}
\newcommand{\C}{\ensuremath{\mathbb{C}}}
\DeclareMathOperator{\Id}{I}

\def\x{x}

\def\f{f}
\def\e{e}

\def\vv{\mathbf{v}}

\def\V{\mathcal{V}}
\def\M{\mathcal{M}}

\DeclareMathOperator{\DFT}{F}
\DeclareMathOperator{\BL}{BL}

\DeclareMathOperator{\Adj}{A}
\DeclareMathOperator{\Cm}{C}
\DeclareMathOperator{\D}{D}

\DeclareMathOperator{\Eig}{\Lambda}
\DeclareMathOperator{\Pj}{P}

\DeclareMathOperator{\Vm}{V}
\DeclareMathOperator{\X}{X}
\DeclareMathOperator{\Um}{U}

\newcommand{\DSPG}{$\mbox{DSP}_{\mbox{\scriptsize G}}$}

\newcommand{\highlightChange}{\color{red}}
\def\HC{\highlightChange}

\begin{document}
\title{ Discrete Signal Processing on Graphs: \\Sampling Theory}
\author{Siheng~Chen, Rohan~Varma,
Aliaksei~Sandryhaila,
  Jelena~Kova\v{c}evi\'c
  \thanks{S. Chen and R. Varma are with the Department of Electrical and Computer
    Engineering, Carnegie Mellon University, Pittsburgh, PA, 15213
    USA. Email: sihengc@andrew.cmu.edu, rohanv@andrew.cmu.edu. A. Sandryhaila is with HP Vertica, Pittsburgh, PA, 15203
    USA. Email: aliaksei.sandryhaila@hp.com. J. Kova\v{c}evi\'c is with the
      Departments of Electrical and
      Computer Engineering and Biomedical Engineering, Carnegie Mellon University, Pittsburgh,
      PA. Email: jelenak@cmu.edu.  }
}
\markboth{IEEE Trans. Signal Process. To appear.} {Chen \MakeLowercase{\textit{et al.}}: Sampling}
 \maketitle


\begin{abstract}
We propose a sampling theory for signals that are supported on either directed or undirected graphs. The theory follows the same paradigm as classical sampling theory. We show that  perfect recovery is possible for graph signals bandlimited under the graph Fourier transform. The sampled signal coefficients form a new graph signal, whose corresponding graph structure preserves the first-order difference of the original graph signal. For general graphs, an optimal sampling operator based on experimentally designed sampling is proposed to guarantee perfect recovery and robustness to noise; for graphs whose graph Fourier transforms are frames with maximal robustness to erasures as well as for Erd\H{o}s-R\'enyi graphs, random sampling leads to perfect recovery with high probability. We further establish the connection to the sampling theory of finite discrete-time signal processing and previous work on signal recovery on graphs. To handle full-band graph signals, we propose a graph filter bank based on sampling theory on graphs.  Finally, we apply the proposed sampling theory to semi-supervised classification on online blogs and digit images, where we achieve similar or better performance with fewer labeled samples compared to previous work.

\end{abstract}

\section{Introduction}
\label{sec:intro}
With the explosive growth of information and communication, signals are generated at an unprecedented rate from various sources, including social, citation, biological, and physical infrastructure~\cite{Jackson:08,Newman:10}, among others. Unlike time-series signals or images, these signals possess complex, irregular structure, which requires novel processing techniques leading to the emerging field of signal processing on graphs~\cite{ShumanNFOV:13,SandryhailaM:14}.

Signal processing on graphs extends classical
discrete signal processing to
signals with an underlying complex, irregular
structure. The framework models that underlying structure by a graph and signals by graph signals, generalizing concepts and tools from classical discrete signal
processing to graph signal processing.  Recent work includes graph-based
filtering~\cite{SandryhailaM:13,NarangO:12,NarangO:13}, graph-based
transforms~\cite{SandryhailaM:13,HammondVG:11,NarangSO:10}, sampling and interpolation
on graphs~\cite{Pesenson:08,NarangGO:13, AnisGO:14}, uncertainty principle on graphs~\cite{AgaskarL:13}, semi-supervised
classification on graphs~\cite{ChenCRBGK:13,ChenSMK:13,EkambaramFAB:13}, graph
dictionary learning~\cite{DongTFV:14, ThanouSF:14}, denoising~\cite{NarangO:12, ChenSMK:14a}, community detection and clustering on graphs~\cite{Tremblay:14, DongFVN:14, ChenO:14}, graph signal recovery~\cite{ChenSMK:14, ChenSLWMRBGK:14, WangLG:14} and distributed algorithms~\cite{WangLG:15a, ChenSK:15b}.

Two basic approaches to signal processing on graphs have been
considered: The first is rooted
in the~\emph{spectral graph theory} and builds upon the~\emph{graph Laplacian matrix}~\cite{ShumanNFOV:13}. Since the standard graph Laplacian matrix is restricted to be symmetric and positive semi-definite, this approach is applicable only to undirected graphs with real and nonnegative edge
weights. The second approach,~\emph{discrete signal processing on graphs}
(\DSPG)~\cite{SandryhailaM:13,SandryhailaM:131}, is rooted in the
\emph{algebraic signal processing theory}~\cite{PueschelM:08,Pueschelm:08b} and builds upon the graph
shift operator, which works as the elementary
operator that generates all linear shift-invariant filters for signals
with a given structure. The graph shift operator is the adjacency matrix and represents the relational dependencies
between each pair of nodes. Since the graph shift is not restricted to be
symmetric, the corresponding framework is applicable to arbitrary
graphs, those with undirected or directed edges, with real or complex,
nonnegative or negative weights. Both frameworks analyze
signals with complex, irregular structure, generalizing a series of
concepts and tools from classical signal processing, such as graph
filters, graph Fourier transform, to diverse graph-based applications.

In this paper, we consider the classical signal processing task of sampling and interpolation within the framework of \DSPG~\cite{VetterliKG:12,KovacevicP:08}. As the bridge connecting sequences and functions,  classical sampling theory shows that a bandlimited function can be perfectly recovered from its sampled sequence if the sampling rate is high enough~\cite{Unser:00}. More generally, we can treat any decrease in dimension via a linear operator as sampling, and, conversely, any increase in dimension via a linear operator as interpolation~\cite{VetterliKG:12,ChenSK:15a}. Formulating a
sampling theory in this context is equivalent to moving between  higher- and lower-dimensional spaces. 

A sampling theory for graphs has interesting applications. For example, given a graph representing friendship connectivity on Facebook, we can  sample a  fraction of users and query their hobbies and then recover all users' hobbies. The task of sampling on graphs is, however, not well understood~\cite{NarangGO:13,AnisGO:14}, because graph signals lie on complex, irregular structures. It is even more challenging to find a graph structure that is associated with the sampled signal coefficients; in the Facebook example, we sample a small fraction of users and an associated graph structure would allow us to infer new connectivity between those sampled users, even when they are not directly connected in the original graph.

Previous works on sampling theory~\cite{Pesenson:08,AnisGO:14,GaddeAO:14} consider graph signals that are uniquely sampled onto a given subset of nodes. This approach is hard to apply to directed graphs. It also does not explain which graph structure supports these sampled coefficients.

In this paper, we propose a sampling theory for signals that are supported on either directed or undirected graphs. Perfect recovery is possible for graph signals bandlimited under the graph Fourier transform. We also show that the sampled signal coefficients form a new graph signal whose corresponding graph structure is constructed from the original graph structure. The proposed sampling theory follows Chapter~5 from~\cite{VetterliKG:12} and is consistent with classical sampling theory. 

We call a sampling operator that leads to perfect recovery~\emph{a qualified sampling operator}. We show that for general graphs, an optimal sampling operator based on experimentally designed sampling is proposed to guarantee perfect recovery and robustness to noise; for graphs whose graph Fourier transforms are frames with maximal robustness to erasures as well as for Erd\H{o}s-R\'enyi graphs, random sampling leads to perfect recovery with high probability. We further establish the connection to sampling theory of finite discrete-time signal processing and previous works on sampling theory on graphs. To handle full-band graph signals, we propose graph filter banks to force graph signals to be bandlimited. Finally, we apply the proposed sampling theory to semi-supervised classification of online blogs and digit images, where we achieve similar or better performance with fewer labeled samples compared to the previous works.

\mypar{Contributions}  The main contributions of the
  paper are as follows:
\begin{itemize}
\item A novel  framework for sampling a graph signal, which solves complex sampling problems by using simple tools from linear algebra;
\item A novel approach for sampling a graph by preserving the first-order difference of the original graph signal;
\item A novel approach for designing a sampling operator on graphs.
\end{itemize}

\mypar{Outline of the paper} Section~\ref{sec:DSPG} formulates the problem and briefly reviews \DSPG, which lays the foundation for this paper; Section~\ref{sec:STG} describes the proposed sampling theory for graph signals, and the proposed construction of graph structures for the sampled signal coefficients; Section~\ref{sec:qualsSO} studies the qualified sampling operator, including random sampling and experimentally designed sampling; Section~\ref{sec:discuss} discusses the relations to previous works and extends the sampling framework to the design of graph filter banks; Section~\ref{sec:apps} shows the application to semi-supervised learning; Section~\ref{sec:conclusions} concludes the paper and provides pointers to future directions.

\section{Discrete Signal Processing on Graphs}
\label{sec:DSPG}
In this section, we briefly review relevant concepts of discrete
signal processing on graphs; a thorough introduction can be found
in~\cite{SandryhailaM:14,SandryhailaM:131}.  It is a theoretical
framework that generalizes classical discrete signal processing from
regular domains, such as lines and rectangular lattices, to irregular
structures that are commonly described by graphs.

\subsection{Graph Shift} 
Discrete signal processing on graphs studies signals with complex,
irregular structure represented by a graph $G = (\V,\Adj)$, where $\V
= \{v_0,\ldots, v_{N-1}\}$ is the set of nodes and $\Adj \in \C^{N
  \times N}$ is the~\emph{graph shift}, or a weighted adjacency matrix. It represents the connections of the graph $G$, which can be
either directed or undirected (note that the standard graph Laplacian matrix
can only represent undirected graphs~\cite{ShumanNFOV:13}). The edge weight $\Adj_{n,m}$ between
nodes $v_n$ and $v_m$ is a quantitative expression of the underlying
relation between the $n$th and the $m$th node, such as similarity, dependency, or a communication pattern. To guarantee that the shift operator is properly scaled, we normalize the graph shift A to satisfy $|\lambda_{\max} (\Adj)|= 1$.

\subsection{Graph Signal}
Given the graph representation $G = (\V,\Adj)$, a \emph{graph signal}
is defined as the map on the graph nodes that assigns the signal
coefficient $x_n\in\C$ to the node $v_n$.  Once the node order is
fixed, the graph signal can be written as a vector
\begin{equation}
\label{eq:graph_signal}
  \x \ = \ \begin{bmatrix}
 x_0 ~ x_1~  \ldots ~x_{N-1}
\end{bmatrix}^T \in \C^N,
\end{equation}
where the $n$th signal coefficient
corresponds to  node $v_n$.

\subsection{Graph Fourier Transform} 
In general, a Fourier transform corresponds to the expansion of a
signal using basis elements that are invariant to filtering; here, this
basis is the eigenbasis of the graph shift $\Adj$ (or, if the complete
eigenbasis does not exist, the Jordan eigenbasis of $\Adj$). For simplicity, assume that $\Adj$ has a complete eigenbasis and the
spectral decomposition of $\Adj$
is~\cite{VetterliKG:12}
\begin{equation}
  \label{eq:eigendecomposition}
  \Adj=\Vm\Eig\Vm^{-1},
\end{equation}
where the eigenvectors of $\Adj$ form the columns of matrix $\Vm$, and
$\Eig\in\C^{N\times N}$ is the diagonal matrix of corresponding
eigenvalues $\lambda_0, \, \ldots, \, \lambda_{N-1}$ of $\Adj$. These eigenvalues represent frequencies on the graph~\cite{SandryhailaM:131}. We do not specify the ordering of graph frequencies here and we will explain why later.
\begin{defn}
\label{df:graph_FT}
The~\emph{graph Fourier transform} of $\x \in \C^N$ is
\begin{equation}
  \label{eq:graph_FT}
  \widehat{\x} = \Vm^{-1} \x.
\end{equation}
 The~\emph{inverse graph Fourier transform} is
 \begin{equation}
 \nonumber
 \x  =  \Vm  \widehat{\x} .
 \end{equation}
\end{defn}
\noindent
The vector $\widehat{\x}$ in~\eqref{eq:graph_FT} represents the
signal's expansion in the eigenvector basis and describes the frequency
content of the graph signal $\x$. The inverse graph Fourier transform
reconstructs the graph signal from its frequency content by combining
graph frequency components weighted by the coefficients of the
signal's graph Fourier transform.

\begin{table}[h]
  \footnotesize
  \begin{center}
    \begin{tabular}{@{}lll@{}}
      \toprule
      {\bf Symbol}  & {\bf Description} & {\bf Dimension}\\
      \midrule \addlinespace[1mm]
      $ \Adj $ &  graph shift &  $N \times N$\\ 
      $ \x $ &  graph signal &  $N$\\
      $ \Vm^{-1}$  & graph Fourier transform matrix &  $N \times N$\\             	  $ \widehat{\x}$ &  graph signal in the frequency domain &  $N$\\
      $\Psi $ & sampling operator &  $M \times N$\\ 
      $\Phi $ & interpolation operator &  $N \times M$\\ 
	  $ \M $ &  sampled indices &  \\
	  $\x_\M $ & sampled signal coeffcients of $\x$&  $M$\\ 
	  $\widehat{\x}_{(K)} $ & first $K$ coeffcients of $\widehat{\x}$&  $K$\\ 
	  $\Vm_{(K)} $ & first $K$ columns of $\Vm$&  $N \times K$\\ 
      \bottomrule
    \end{tabular}
  \end{center}
  \caption{\label{table:parameters}
    Key notation used in the paper{\HC .}
 }
\end{table}

\section{Sampling on Graphs}
\label{sec:STG}
Previous works on sampling theory of graph signals is based on spectral graph theory~\cite{AnisGO:14}. The bandwidth of graph signals is defined based on the value of graph frequencies, which correspond to the eigenvalues of the graph Laplacian matrix. Since each graph has its own graph frequencies, it is hard in practice to specify a general cut-off graph frequency; it is also computationally inefficient to compute all the values of graph frequencies, especially for large graphs.

In this section, we propose a novel sampling framework for graph signals. Here, the bandwidth definition is  based on the number of non-zero signal coefficients in the graph Fourier domain. Since each signal coefficient in the graph Fourier domain corresponds to a graph frequency, the bandwidth definition is also based on the number of graph frequencies. This makes the proposed sampling framework strongly connected to linear algebra, that is, we are allowed to use simple tools from linear algebra to perform sampling on complex, irregular graphs.

\subsection{Sampling \& Interpolation}

Suppose that we want to sample $M$ coefficients of a graph signal $\x \in
\C^N$ to produce a sampled signal $\x_\M \in \C^M$ ($M < N$), where $\M = (\M_0, \cdots, \M_{M-1})$ denotes the sequence of \emph{sampled} indices, and $\M_i \in \{0, 1, \cdots, N-1 \}$. We then interpolate $\x_{\M}$ to get
$\x' \in \C^N$, which recovers $\x$ either exactly or
approximately. The sampling operator $\Psi$ is a linear mapping from
$\C^N$ to $\C^M$, defined as
\begin{equation}
\label{eq:Psi}
 \Psi_{i,j} = 
  \left\{ 
    \begin{array}{rl}
      1, & j = \M_i;\\
      0, & \mbox{otherwise},
  \end{array} \right. 
\end{equation}
and the interpolation operator $\Phi$ is a linear
mapping from $\C^M$ to $\C^N$ (see Figure~\ref{fig:sampling}),
\begin{figure}[t]
  \begin{center}
     \includegraphics[width= 0.95\columnwidth]{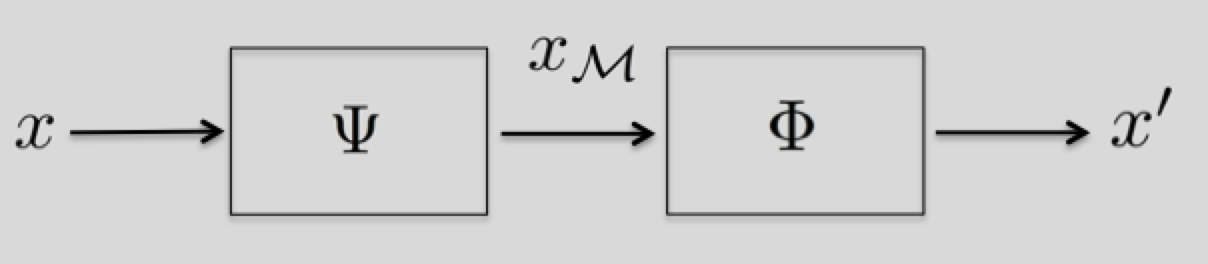}
  \end{center}
  \caption{\label{fig:sampling}  Sampling followed by interpolation.}
  \vspace{-0.15in}
\end{figure}

\begin{eqnarray}
{\rm sampling:}~~&&\x_{\M} =  \Psi \x \in \C^{M},
\nonumber
\\ \nonumber
{\rm interpolation:}~~&&\x' =  {\Phi} \x_{\M} = \Phi \Psi \x  \in \C^{N},
\end{eqnarray}
where $\x' \in \R^N$ recovers $\x$ either exactly or approximately.  We consider two sampling strategies:~\emph{random sampling} means that sample indices are chosen from $\{0, 1, \cdots, N-1\}$ independently and randomly; and~\emph{experimentally design sampling} means that sample indices can be chosen beforehand. It is clear that random sampling is a subset of experimentally design~sampling.

Perfect recovery happens for all $\x$ when $\Phi \Psi$ is the identity matrix. This
is not possible in general because ${\rm rank}(\Phi \Psi) \leq M <
N$; it is, however, possible to do this for signals with specific structure that we will define as bandlimited graph signals, as in classical discrete signal processing.

\subsection{Sampling Theory for Graph Signals}
\label{sec:ggs}
 We now define a class of bandlimited graph signals, which makes perfect recovery possible.

\begin{defn}
  \label{df:GBL}
  A graph signal is called~\emph{bandlimited} when there exists a $K \in
  \{0, 1, \cdots, N-1\}$ such that its graph Fourier transform
  $\widehat{\x}$ satisfies
  \begin{displaymath}
  \widehat{x}_k  =  0 \quad {\rm for~all~}  \quad k \geq K.
\end{displaymath}
The smallest such $K$ is called the~\emph{bandwidth} of $\x$. A graph
signal that is not bandlimited is called a \emph{full-band graph signal}.
\end{defn}
Note that the bandlimited graph signals here do not necessarily mean low-pass, or smooth. Since we do not specify the ordering of frequencies, we can reorder the eigenvalues and permute the corresponding eigenvectors in the graph Fourier transform matrix to choose any band in the graph Fourier domain. The bandlimited graph signals are smooth only when we sort the eigenvalues in a descending order. The bandlimited restriction here is equivalent to limiting the number of non-zero signal coefficients in the graph Fourier domain with known supports. This generalization is potentially useful to represent non-smooth graph signals.

\begin{defn}
  \label{df:GBLS}
  The set of graph signals in $\C^N$ with bandwidth of at most $K$ is
  a closed subspace denoted $\BL_K(\Vm^{-1})$, with $\Vm^{-1}$ as
  in~\eqref{eq:eigendecomposition}.
\end{defn}

When defining the bandwidth, we focus on the number of graph frequencies, while previous works~\cite{AnisGO:14} focus on the value of graph frequencies.
There are two shortcomings to using the values of graph frequencies: (a) When considering the values of graph frequencies, we ignore the discrete nature of graphs; because graph frequencies are discrete, two cut-off graph frequencies on the same graph can lead to the same bandlimited space. For example, assume a graph has graph frequencies 0, 0.1, 0.4, 0.6 and 2; when we set the cut-off frequency to either 0.2 or 0.3, they lead to the same bandlimited space; (b) The values of graph frequencies cannot be compared between different graphs. Since each graph has its own graph frequencies, a same value of the cut-off graph frequency on two graphs can mean different things. For example, one graph has graph frequencies as 0, 0.1, 0.2, 0.4 and 2, and another has graph frequencies 0, 1.1, 1.6, 1.8, and 2; when we set the cut-off frequency to 1, that is, we preserve all the graph frequencies that are no greater than 1, first graph preserves three out of four graph frequencies and the second graph only preserves one out of four. The values of graph frequencies thus do not necessarily give a direct and intuitive understanding about the bandlimited space. Another key advantage of using the number of graph frequencies is to build a strong connection to linear algebra allowing for the use of simple tools from linear algebra in sampling and interpolation of bandlimited graph signals.

In Theorem 5.2 in~\cite{VetterliKG:12}, the authors show the recovery for vectors via projection, which lays the theoretical foundation for the classical sampling theory. Following the theorem, we obtain the following result, the proof of which can be found in~\cite{ChenSK:15a}.
\begin{myThm}
  \label{thm:GPR}
  Let $\Psi$ satisfy
  \begin{equation}
  \nonumber
  \label{eq:Assumption}
    {\rm rank}( \Psi \Vm_{(K)}) = K,
  \end{equation}
  where $\Vm_{(K)} \in \R^{N \times K}$ denotes the first $K$ columns of $\Vm$. For all $\x \in \BL_K(\Vm^{-1})$, perfect recovery, $\x = \Phi \Psi \x$, is achieved by choosing 
  \begin{equation}
  \nonumber
  \label{eq:interpolation}
  \Phi = \Vm_{(K)} \Um,
  \end{equation}
with $\Um \Psi
  \Vm_{(K)}$ a $K \times K$ identity matrix.
\end{myThm}
Theorem~\ref{thm:GPR} is applicable for all graph signals that have a few non-zero elements in the graph Fourier domain with known supports, that is, $K < N$.

Similarly to the classical sampling theory, the sampling rate has a lower bound for graph signals as well, that is, the sample size $M$ should be no smaller than the bandwidth $K$. When $M < K$, rank$(\Um \Psi
  \Vm_{(K)}) \leq$ rank$(\Um) \leq M < K$, and thus, $\Um \Psi
  \Vm_{(K)}$ can never be an identity matrix. For $\Um \Psi  \Vm_{(K)}$ to be an identity matrix,  $\Um$ is the inverse of $\Psi  \Vm_{(K)}$ when $M = K$; it is a pseudo-inverse of $\Psi  \Vm_{(K)}$ when $M > K$, where the redundancy can be useful for reducing the influence of noise. For simplicity, we only consider $M = K$ and $\Um$ invertible. When $M > K$, we simply select $K$ out of $M$ sampled signal coefficients to ensure that the sample size and the bandwidth are the same.

From Theorem~\ref{thm:GPR}, we see that an arbitrary sampling operator
may not lead to perfect recovery even for bandlimited graph signals. When the sampling operator $\Psi$ satisfies the full-rank assumption~\eqref{eq:Assumption}, we call it a~\emph{qualified sampling operator}. To satisfy~\eqref{eq:Assumption}, the sampling operator should select at least one set of $K$
linearly-independent rows in $\Vm_{(K)}$. Since $\Vm$ is invertible, the column vectors in $\Vm$ are linearly
independent and rank$(\Vm_{(K)}) = K$ always holds; in other words, at
least one set of $K$ linearly-independent rows in $\Vm_{(K)}$ always
exists.  Since the graph shift $\Adj$ is given, one can find such a set
independently of the graph signal. Given such a set,
Theorem~\ref{thm:GPR} guarantees perfect recovery of bandlimited graph
signals. To find linearly-independent
rows in a matrix, fast algorithms exist, such as QR decomposition;
see~\cite{HornJ:85,VetterliKG:12}. Since we only need to know the graph structure to design a~\emph{qualified sampling operator}, this follows the experimentally designed sampling. We will expand this topic in Section~\ref{sec:qualsSO}.

\subsection{Sampled Graph Signal}
We just showed that perfect recovery is possible when the graph signal is bandlimited.  We now show that the sampled signal coefficients form a new graph signal, whose corresponding graph shift can be constructed from the original graph shift.

Although the following results can be generalized to $M > K$ easily, we only consider $M = K$ for simplicity. Let the sampling operator $\Psi$ and the interpolation operator~$\Phi$ satisfy the conditions in Theorem~\ref{thm:GPR}. For all $\x \in \BL_{K} (\Vm^{-1})$, we have
\begin{eqnarray}
\label{eq:recovery}
\x \ = \ \Phi \Psi \x  \ = \  \Phi \x_\M & \stackrel{(a)}{=} & \Vm_{(K)} \Um \x_\M
\nonumber \\ \nonumber
& \stackrel{(b)}{=}  & \Vm_{(K)}  \widehat{\x}_{(K)},
\end{eqnarray}
where $\widehat{\x}_{(K)}$ denotes the first $K$ coefficients of $\widehat{\x}$, (a) follows from Theorem~\ref{thm:GPR} and (b) from Definition~\ref{df:GBL}. We thus get
\begin{equation}
\nonumber
\widehat{\x}_{(K)} \ = \ \Um \x_\M,
\end{equation}
 and
\begin{eqnarray}
\nonumber
\x_\M \ = \  \Um^{-1}  \Um \x_\M  =\Um^{-1}  \widehat{\x}_{(K)}.
\end{eqnarray}
From what we have seen, the sampled signal coefficients $x_\M$ and the frequency content $\widehat{\x}_{(K)}$ form a Fourier pair because  $\x_\M$ can be constructed from $ \widehat{\x}_{(K)}$ through $\Um^{-1}$ and $ \widehat{\x}_{(K)}$ can also be constructed from $\x_\M$ through $\Um$. This implies that, according to Definition~\ref{df:graph_FT} and the spectral decomposition~\eqref{eq:eigendecomposition}, $\x_\M$ is a graph signal associated with the graph Fourier transform matrix $\Um$ and  a new graph shift
\begin{displaymath}
\Adj_\M = \Um^{-1} \Lambda_{(K)} \Um \in \C^{K \times K},
\end{displaymath}
where $\Lambda_{(K)} \in \C^{K \times K}$ is a diagonal matrix that samples the first $K$ eigenvalues of $\Lambda$. This leads to the following theorem.
\begin{myThm}
\label{thm:sg}
Let $\x \in \BL_{K}(\Vm^{-1})$ and let
\begin{displaymath}
\x_\M =  \Psi x \in \C^K
\end{displaymath}
 be its sampled version, where $\Psi$ is a qualified sampling operator. Then,  the graph shift associated with the graph signal $\x_\M$ is
\begin{equation}
  \Adj_\M = \Um^{-1} \Lambda_{(K)} \Um \in \C^{K \times K},
\end{equation}
with $\Um = (\Psi\Vm_{(K)})^{-1}$.
\end{myThm}
From Theorem~\ref{thm:sg}, we see that the graph shift $\Adj_\M$ is constructed by sampling the rows of the eigenvector matrix and sampling the first $K$ eigenvalues of the original graph shift $\Adj$. We simply say that $\Adj_\M$ is sampled from $\Adj$, preserving certain information in the graph Fourier domain.

Since the bandwidth of $\x$ is $K$, the first $K$ coefficients in the frequency domain are $\widehat{\x}_{(K)} = \widehat{\x}_\M$, and the other $N-K$ coefficients are $\widehat{\x}_{(-K)} = 0$; in other words, the frequency contents of the original graph signal $\x$ and the sampled graph signal $\x_\M$ are equivalent after performing their corresponding  graph Fourier transforms. 

Similarly to Theorem~\ref{thm:GPR}, by reordering the eigenvalues and permuting the corresponding eigenvectors in the graph Fourier transform matrix, Theorem~\ref{thm:sg} is applicable to all graph signals that have limited support in the graph Fourier domain.

\begin{figure*}[t]
  \begin{center}
     \includegraphics[width= 1.2\columnwidth]{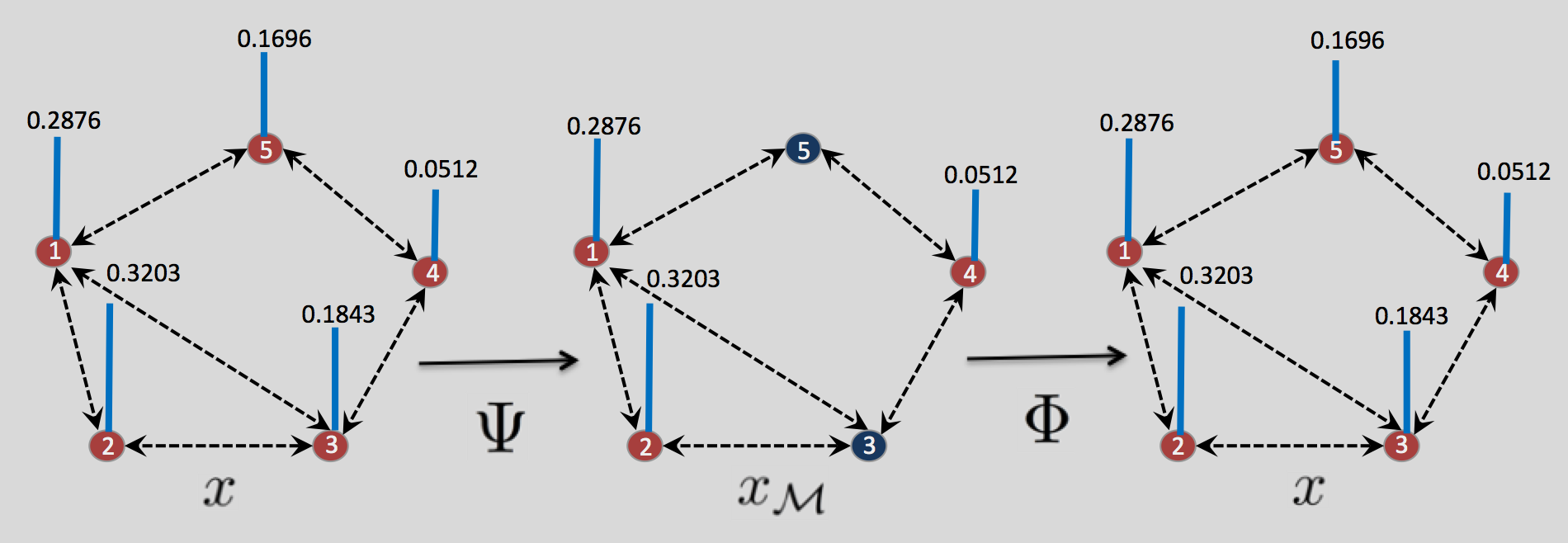}
  \end{center}
  \caption{\label{fig:sampling_signal}  Sampling followed by interpolation. The arrows indicate that the edges are directed.}
  \vspace{-0.15in}
\end{figure*}

\subsection{Property of A Sampled Graph Signal}
We argued that $ \Adj_\M = \Um^{-1} \Lambda_{(K)} \Um$ is the graph shift that supports the sampled signal coefficients $\x_\M$ following from a mathematical equivalence between the graph Fourier transform for the sampled graph signal and the graph shift. We, in fact, implicitly proposed an approach to sampling graphs. Since sampled graphs always lose information, we now study which information $\Adj_\M$ preserves.

\begin{myThm}
\label{thm:sg_prop}
For all $\x \in \BL_{K}(\Vm^{-1})$,
$$ \x_\M - \Adj_\M \x_\M \ = \  \Psi  \left( x - \Adj x \right).
$$ 
\end{myThm}

\begin{proof}
\begin{eqnarray*}
 \x_\M - \Adj_\M \x_\M & = &   \Um^{-1}  \widehat{\x}_\M -  \Um^{-1} \Lambda_{(K)} \Um  \Um^{-1}  \widehat{\x}_\M
 \\
 & = & \Psi  \Vm_{(K)}  (\Id - \Lambda_{(K)} ) \widehat{\x}_{(K)}
 \\
 & = & \Psi  \left( x - \Adj x \right),
\end{eqnarray*}
where the last equality follows from $\x \in \BL_{K}(\Vm^{-1})$.
\end{proof}
The term $x - \Adj x$ measures the difference between the original graph signal and its shifted version. This is also called the first-order difference of $\x$, while the term  $\Psi  \left( x - \Adj x \right)$ measures  the first-order difference of $\x$ at sampled indices. Furthermore, $\left\|\x - \Adj \x \right\|_p^p$ is the graph total variation based on the~\emph{$\ell_p$-norm}, which is a quantitative characteristic that measures the smoothness of a graph signal~\cite{SandryhailaM:131}. When using a sampled graph to represent the sampled signal coefficients, we lose the connectivity information between the sampled nodes and all the other nodes; despite this,  $\Adj_\M$ still preserves the first-order difference at sampled indices. Instead of focusing on preserving connectivity properties as in prior work~\cite{SpielmanS:11}, we emphasize the interplay between signals and structures.

\subsection{Example}
\label{sec:toy}
We consider a five-node directed graph with  graph shift
$$
\Adj = \left[  \begin{array}{llllll}
  0 &  \frac{2}{5} & \frac{2}{5} &  0 & \frac{1}{5} \\
  \frac{2}{3} &  0 & \frac{1}{3} & 0 & 0\\
  \frac{1}{2} &  \frac{1}{4} & 0 & \frac{1}{4} & 0\\
   0 & 0 & \frac{1}{2} & 0 &  \frac{1}{2} \\
  \frac{1}{2} & 0 &  0 & \frac{1}{2} & 0
\end{array} \right].
$$
The corresponding inverse graph Fourier transform matrix is
$$
\Vm =  \left[  \begin{array}{llllll}
0.45 &  \phantom{+}0.19 &  \phantom{+}0.25 &  \phantom{+}0.35 & -0.40 \\
0.45 &  \phantom{+}0.40 &  \phantom{+}0.16 & -0.74 & \phantom{+}0.18 \\ 
0.45 &  \phantom{+}0.08 &  -0.56 & \phantom{+}0.29 & \phantom{+}0.36\\ 
0.45 &  -0.66 & -0.41 & -0.47  &  -0.57 \\
0.45 &  -0.60 &  \phantom{+}0.66 & \phantom{+}0.13  &  \phantom{+}0.59 
\end{array} \right],
$$
and the frequencies are
$$
\Lambda =  {\rm diag} \begin{bmatrix}
		1 &  0.39 & -0.12 & -0.44 & -0.83
 	\end{bmatrix}.
$$
Let $K = 3$; generate a bandlimited graph signal $\x \in \BL_3(\Vm^{-1})$ as 
$$
\widehat{\x} = \begin{bmatrix}  0.5 & 0.2 & 0.1 & 0 & 0 \end{bmatrix}^T,
$$
with
$$
\x = \begin{bmatrix}  0.29 & 0.32 & 0.18 & 0.05 & 0.17 \end{bmatrix}^T,
$$
and the first-order difference of $x$ is
$$
\x - \Adj \x =  \begin{bmatrix}  0.05 & 0.07 & -0.05 & -0.13 & 0.0002 \end{bmatrix}^T.
$$
We can check the first three columns of $\Vm$ to see that all sets of
three rows are independent. According to the sampling theorem, we can then
recover $\x$ perfectly by sampling any three of its coefficients;
for example, sample the first, second and the fourth coefficients. Then, $\M = (1,2,4)$, $\x_\M
= \begin{bmatrix} 0.29 & 0.32 & 0.05  \end{bmatrix}^T$, and the sampling
operator
$$
\Psi = \begin{bmatrix}
1 &  0 & 0 & 0 & 0\\
0 &  1 & 0 & 0 & 0\\
0 &  0 & 0 & 1 & 0
\end{bmatrix}
$$
is qualified. We recover $\x$ by using the following interpolation operator (see Figure~\ref{fig:sampling_signal})
$$
\Phi = \Vm_{(3)} (\Psi \Vm_{(3)})^{-1}
=
\left[  \begin{array}{llllll}
\phantom{+}1 &  \phantom{+}0  & \phantom{+}0 \\
\phantom{+}0 &  \phantom{+}1  & \phantom{+}0  \\
-2.7  & \phantom{+}2.87  & \phantom{+}0.83 \\
\phantom{+}0 & \phantom{+}0 & \phantom{+}1 \\
\phantom{+}5.04 & -3.98  & -0.05
\end{array} \right].
$$
The inverse graph Fourier transform matrix for the sampled signal is 
\begin{displaymath}
\Um^{-1} = \Psi \Vm_{(3)} = 
\left[  \begin{array}{llllll}
0.45 &  \phantom{+}0.19 & \phantom{+}0.25  \\
0.45  &  \phantom{+}0.40  & \phantom{+}0.16 \\ 
0.45  &  -0.66  & -0.41 
\end{array} \right],
\end{displaymath}
and the sampled frequencies are
\begin{displaymath}
\Lambda_{(3)} = 
\left[  \begin{array}{llllll}
1 &  0  & \phantom{+}0\\
0 &  0.39 & \phantom{+}0 \\
0 &  0 &  -0.12
\end{array} \right].
\end{displaymath}
The sampled graph shift is then constructed as
\begin{displaymath}
\Adj_\M = \Um^{-1} \Lambda_{(3)} \Um = 
\left[  \begin{array}{llllll}
\phantom{-}0.07  &  \phantom{+}0.75  & \phantom{+}0.32 \\
-0.23  &  \phantom{+}0.96  & \phantom{+}0.28 \\
\phantom{+}1.17  &  -0.56  & \phantom{+}0.39
\end{array} \right].
\end{displaymath}
The first-order difference of $\x_\M$ is
$$
\x_\M - \Adj_\M \x_\M =  \begin{bmatrix}  0.05 & 0.07 & -0.13 \end{bmatrix}^T \ = \  \Psi (\x - \Adj \x).
$$
We see that the sampled graph shift contains self-loops and negative weights and  seems to be dissimilar to $\Adj$, but $\Adj_\M$ preserves a part of the frequency content of $\Adj$ because $\Um^{-1}$ is sampled from $\Vm$ and $\Lambda_{(3)}$ is sampled from $\Adj$. $\Adj_\M$ also preserves the first-order difference of $x$, which validates Theorem~\ref{thm:sg_prop}.

\begin{figure}[t]
  \begin{center}
     \includegraphics[width= 0.95\columnwidth]{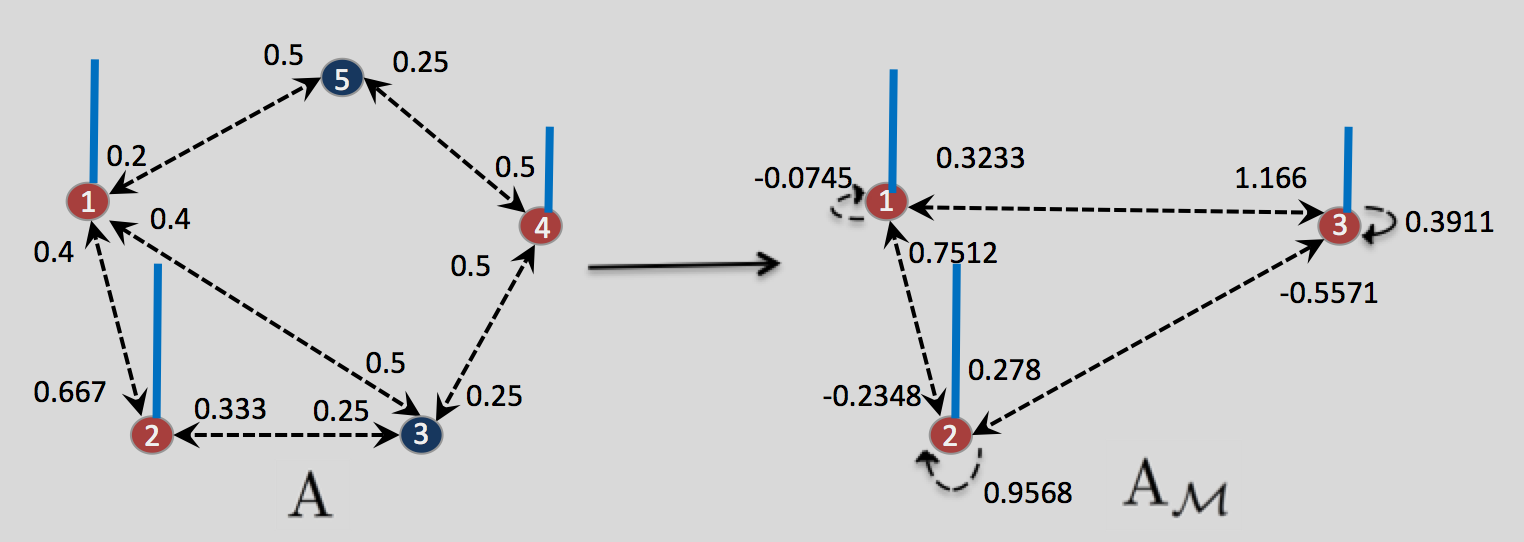}
  \end{center}
  \caption{\label{fig:sampling_graph}  Sampling a graph.}
  \vspace{-0.15in}
\end{figure}

\section{Sampling with A Qualified Sampling Operator}
\label{sec:qualsSO}
As shown in Section~\ref{sec:STG}, only a qualified sampling operator~\eqref{eq:Assumption} can lead to perfect recovery for bandlimited graph signals.  Since a qualified sampling operator~\eqref{eq:Assumption} is designed via the graph structure, it belongs to experimentally designed sampling. The design consist in finding $K$ linearly independent rows in $\Vm_{(K)}$, which gives multiple choices. In this section, we propose an optimal approach to designing a qualified sampling operators by minimizing the effect of noise for general graphs. We then show that for some specific graphs, random sampling also leads to perfect recovery with high probability.

\subsection{Experimentally Designed Sampling}
We now show how to design a qualified sampling operator on any given graph that is robust to noise. We then compare this optimal sampling operator with a random sampling operator on a sensor network.

\subsubsection{Optimal Sampling Operator}
As mentioned in Section~\ref{sec:ggs}, at least one set of $K$ linearly-independent rows in $\Vm_{(K)}$ always exists.
When we have multiple choices of $K$ linearly-independent rows, we aim to find the optimal one to minimize the effect of noise.


We consider a model where noise $\e$ is introduced during sampling as follows,
\begin{eqnarray}
\nonumber
 \x_\M & = & \Psi \x + \e,
 \end{eqnarray}
where $\Psi$ is a qualified sampling operator. The recovered graph signal, $\x_e'$, is then
 \begin{eqnarray}
 \nonumber
  \x_e' & = & \Phi \x_\M  \ = \ \Phi \Psi \x + \Phi \e \ = \ \x + \Phi \e.
 \end{eqnarray}
To bound the effect of noise, we have
\begin{eqnarray*}
 \left\| \x' - \x \right\|_2 & = &   \left\| \Phi \e  \right\|_2 
\ = \  \left\| \Vm_{(K)} \Um \e  \right\|_2
\\
& \leq &  \left\| \Vm_{(K)} \right\|_2  \left\| |\Um  \right\|_2  \left\| \e \right\|_2,
\end{eqnarray*}
where the inequality follows from the definition of the spectral norm. Since $\left\|\Vm_{(K)}  \right\|_2$ and $\left\|\e  \right\|_2$ are fixed, we want $\Um$ to have a small spectral norm. From this perspective, for each feasible $\Psi$, we compute the inverse or pseudo-inverse of $\Psi \Vm_{(K)}$ to obtain $\Um$; the best choice comes from the $\Um$ with the smallest spectral norm. This is equivalent to maximizing the smallest singular value of $\Psi \Vm_{(K)}$, 
\begin{equation}
\label{eq:optimalset}
\Psi^{opt} \ = \ \arg \max_{\Psi}  \sigma_{\min} (\Psi \Vm_{(K)} ),
\end{equation}
where $\sigma_{\min}$ denotes the smallest singular value. The solution of~\eqref{eq:optimalset} is optimal in terms of minimizing the effect of noise; we simply call it~\emph{optimal sampling operator}. Since we restrict the form of $\Psi$ in  \eqref{eq:Psi},~\eqref{eq:optimalset} is non-deterministic polynomial-time hard. To solve~\eqref{eq:optimalset}, we can use a greedy algorithm as shown in Algorithm~\ref{alg:optimalset}. In a previous work, the authors solved a similar optimization problem for matrix approximation and showed that the greedy algorithm gives a good approximation to the global optimum~\cite{AvronB:13}. Note that $\M$ is the sampling sequence, indicating which rows to select, and  $(\Vm_{(K)})_\M$ denotes the sampled rows from $\Vm_{(K)}$.  When increasing the number of samples, the smallest singular value of  $\Psi \Vm_{(K)}$ grows, and thus, redundant samples make the algorithm robust to noise.  

\begin{algorithm}[h]
  \footnotesize
  \caption{\label{alg:optimalset} Optimal Sampling Operator via Greedy Algorithm}
  \begin{tabular}{@{}lll@{}}
    \addlinespace[1mm]
   {\bf Input} 
      & $\Vm_{(K)}$~~the first $K$ columns of $\Vm$ \\
      & $M$~~~~~~the number of samples \\
     {\bf Output}  
      & $\M$~~~~~~sampling set \\
    \addlinespace[2mm]
    {\bf Function} & &\\
    & while $|\M| < M$ \\ 
    &~~~$m \ = \ \arg \max_{i}  \sigma_{\min} \left( (\Vm_{(K)})_{\M + \{ i\}} \right)$ \\
    &~~~$\M \leftarrow \M + \{ m\} $\\
    & end \\
    & {\bf return} $\M$ \\  
     \addlinespace[1mm]
  \end{tabular}
\end{algorithm}

\subsubsection{Simulations}
\label{sec:sensor_optimal}
We consider 150 weather stations in the United States that record local temperatures~\cite{SandryhailaM:13}. The graph
representing these weather stations is obtained by assigning an edge when the geodesic distance between each pair of weather stations is smaller than 500 miles, that is, the graph shift $\Adj$ is formed as 
\begin{displaymath}
\Adj_{i,j} = \begin{cases} 1,&\mbox{when }  0 < d_{i,j} < 500; \\ 
0, & \mbox{otherwise}, \end{cases} 
\end{displaymath}
where $d_{i,j}$ is the geodesic distance between the $i$th and the
$j$th weather stations. 

We simulate a graph signal with bandwidth of 3 as
\begin{displaymath}
\x = \vv_1 + 0.5 \vv_2 + 2 \vv_3,
\end{displaymath}
where $\vv_i$ is the $i$th column in $\Vm$. We design two sampling operators to recover $x$: an arbitrary qualified sampling operator and the optimal sampling operator. We know that both of them recover $x$ perfectly given 3 samples.  To validate the robustness to noise, we add Gaussian noise with mean zero and variance $0.01$ to each sample. We recover the graph signal from its samples by using the interpolation operator~\eqref{eq:interpolation}.

 Figure~\ref{fig:geo_graph_signals} (f) and (h) show the recovered graph signal from each of these two sampling operators. We see that an arbitrary qualified sampling operator is not robust to noise and fails to recover the original graph signal, while the optimal sampling operator is robust to noise and approximately recovers the original graph signal.

\begin{figure*}[htb]
  \begin{center}
    \begin{tabular}{cccc}
\includegraphics[width=0.6\columnwidth]{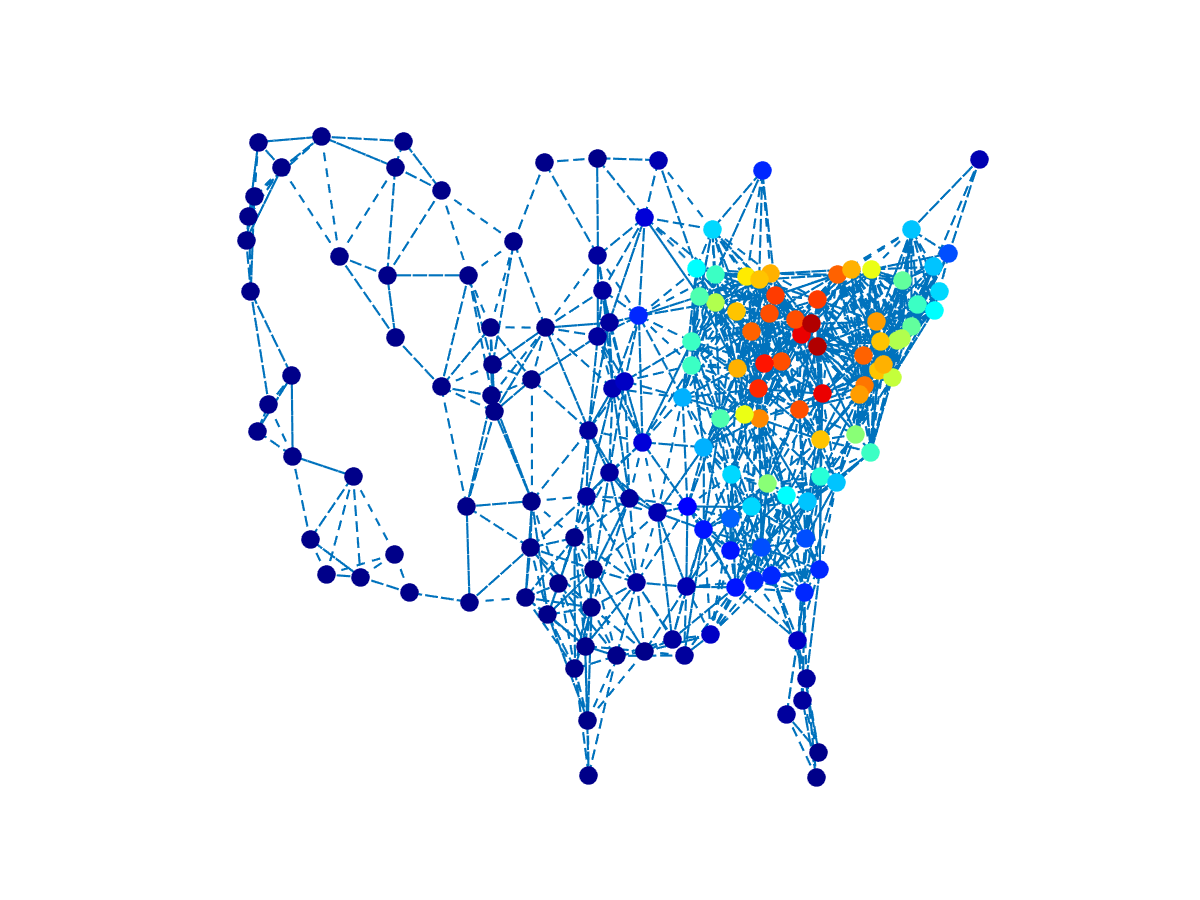}  & \includegraphics[width=0.6\columnwidth]{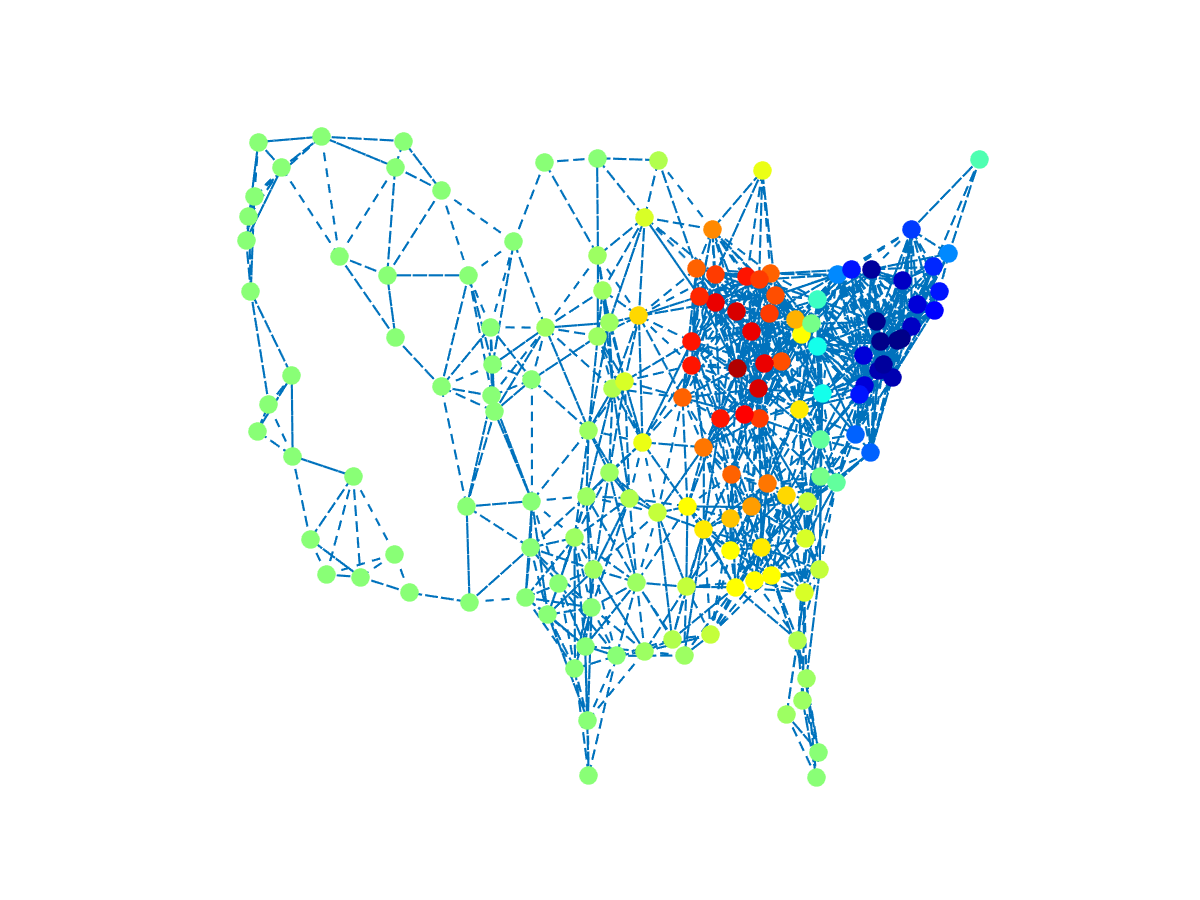}  &
\includegraphics[width=0.6\columnwidth]{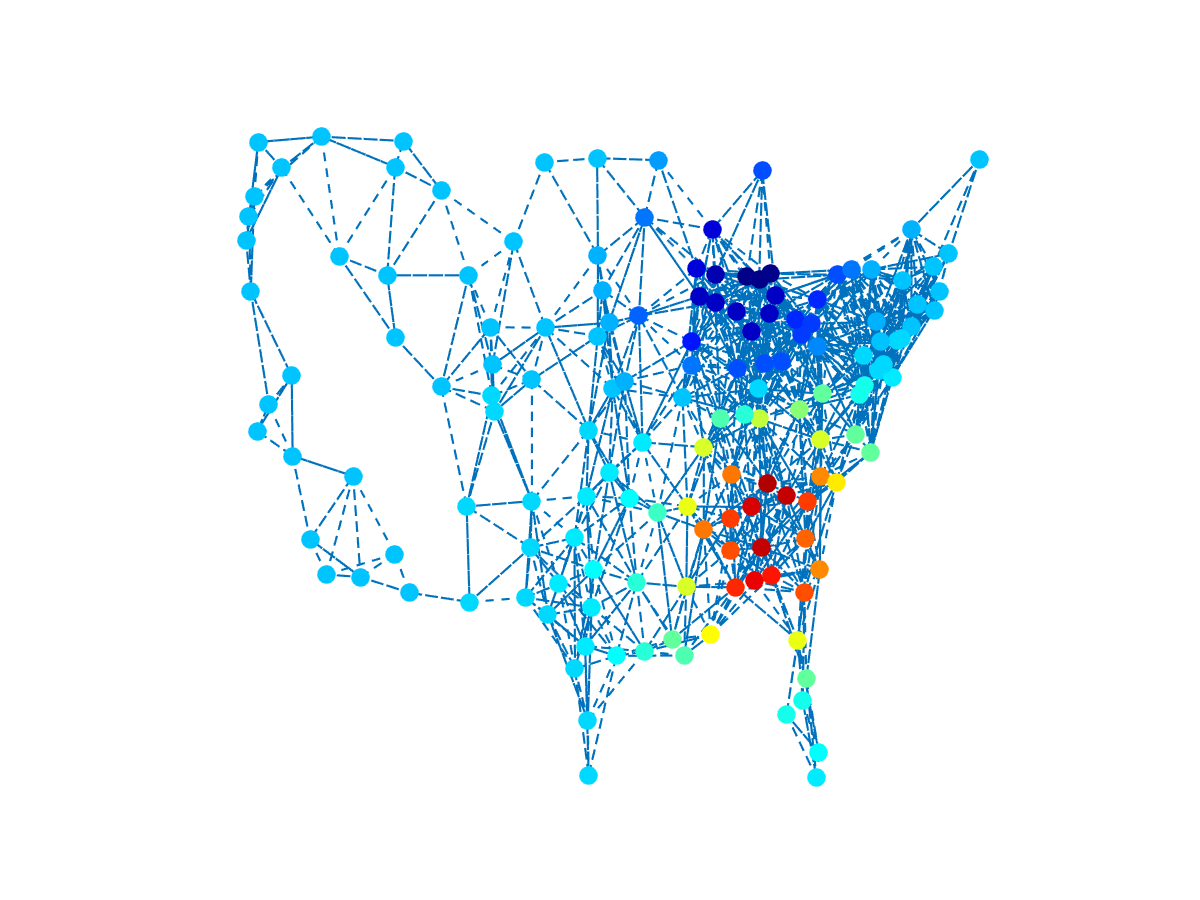} 
\\
 {\small (a) 1st basis vector $\vv_1$.} & {\small (b) 2nd basis vector $\vv_2$.} & {\small (c) 3rd basis vector$\vv_3$.}
 \\
\includegraphics[width=0.6\columnwidth]{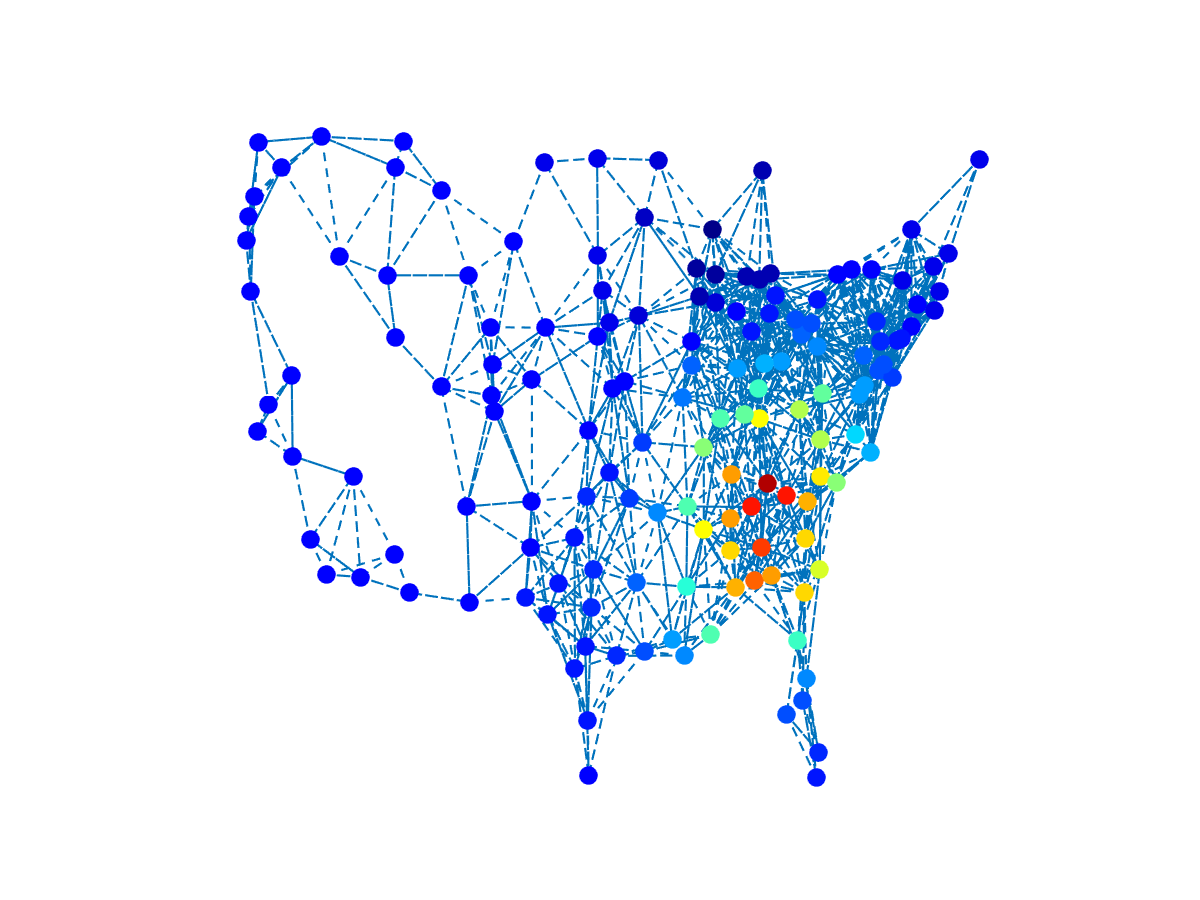}  &
\includegraphics[width=0.6\columnwidth]{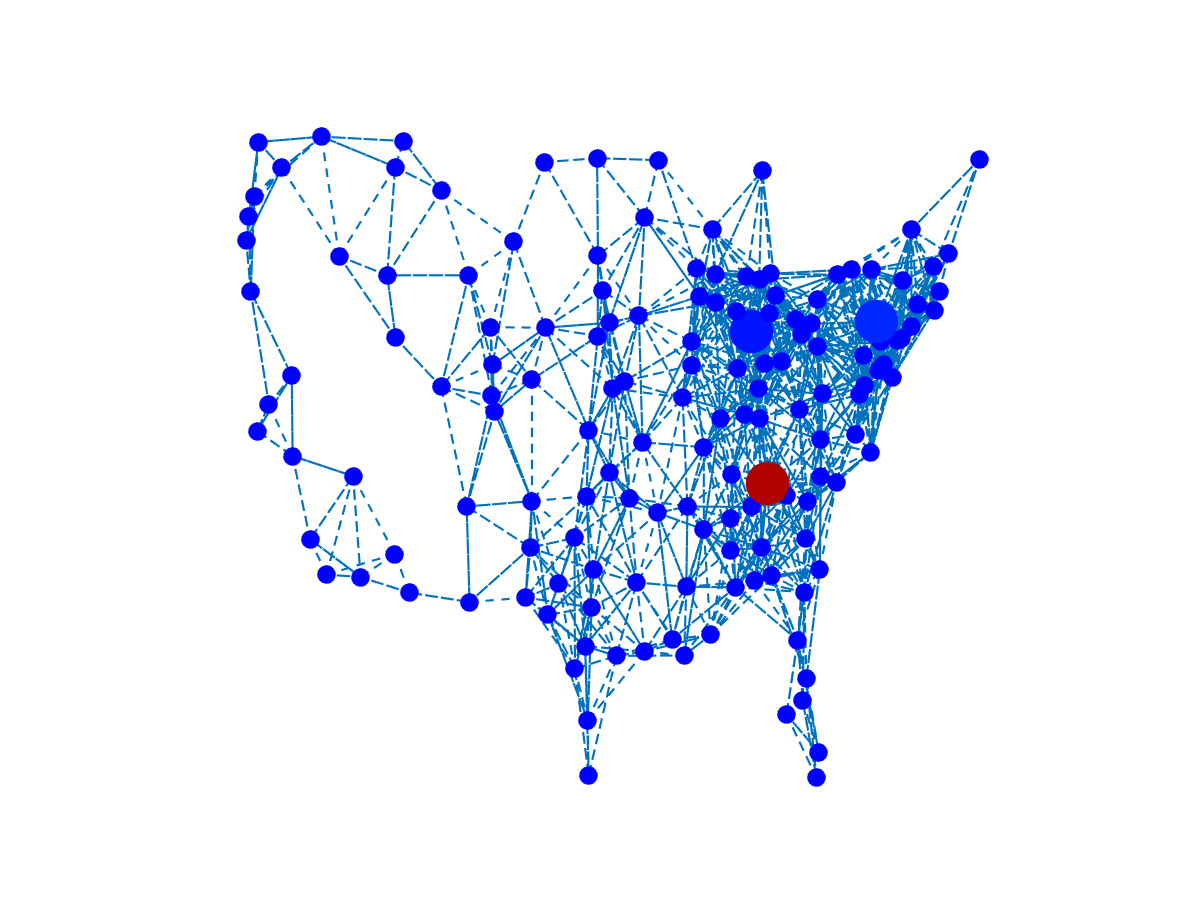}   & \includegraphics[width=0.6\columnwidth]{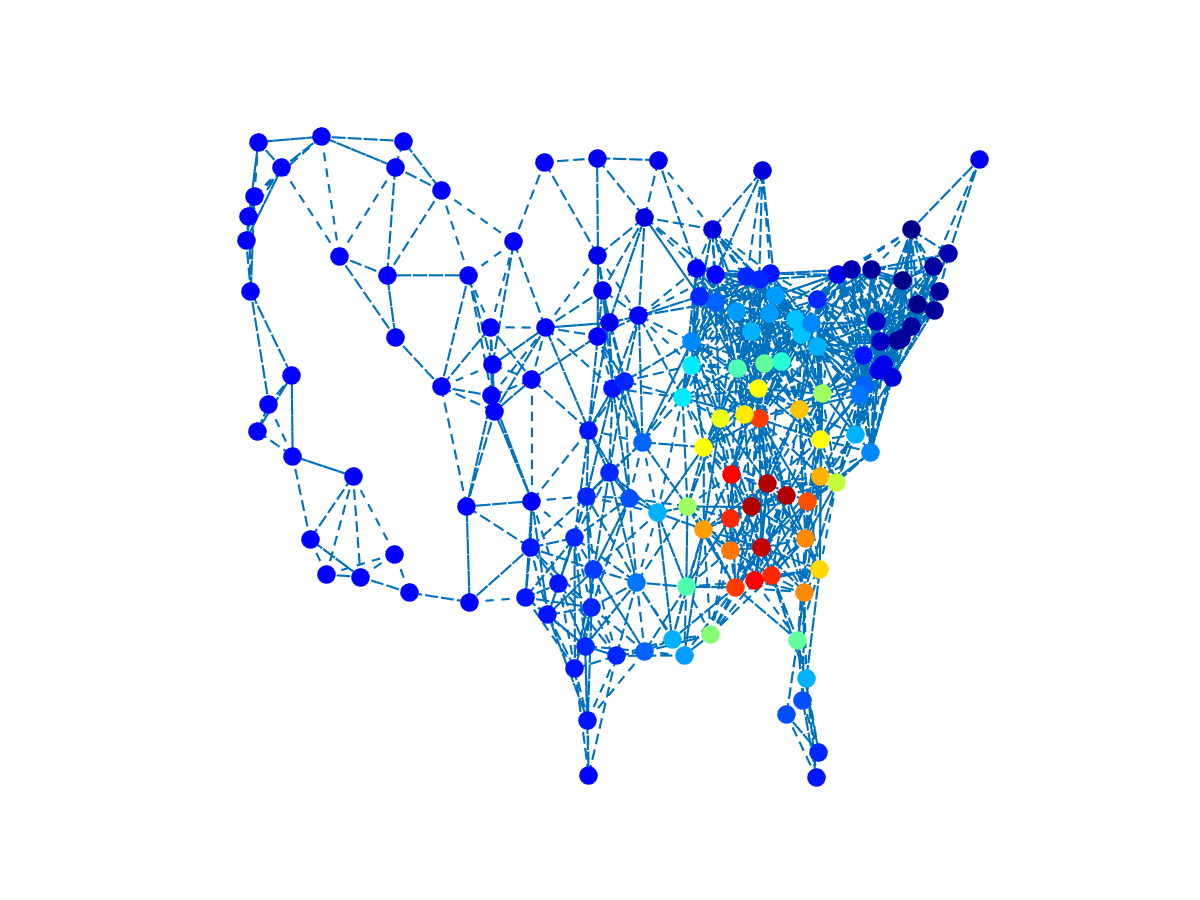}  
   
\\
 {\small (d) Original graph signal.} & {\small (e) Samples from  the optimal } & {\small (f) Recovery from the optimal}
 \\
 & {\small  sampling operator.} & {\small sampling operator.}
 \\  
 \includegraphics[width=0.6\columnwidth]{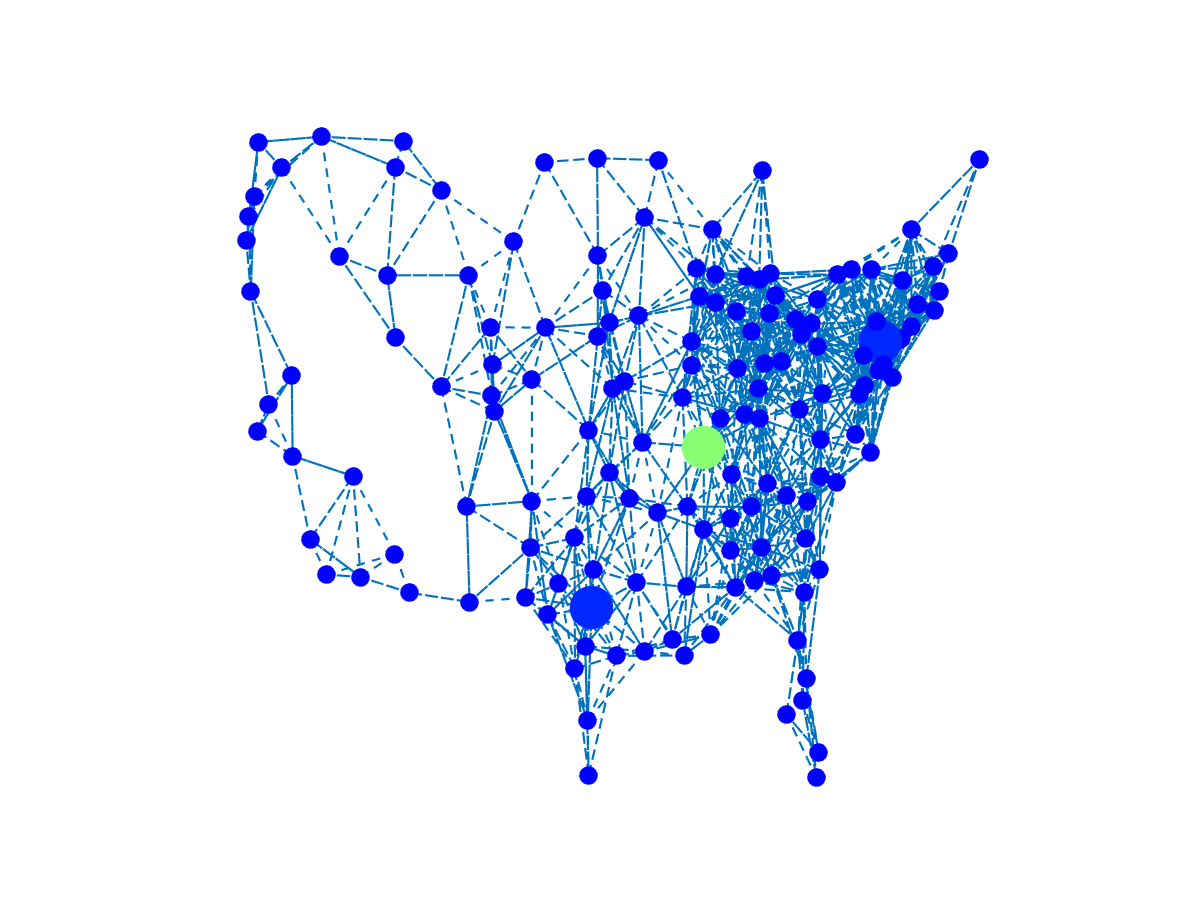}  &
\includegraphics[width=0.6\columnwidth]{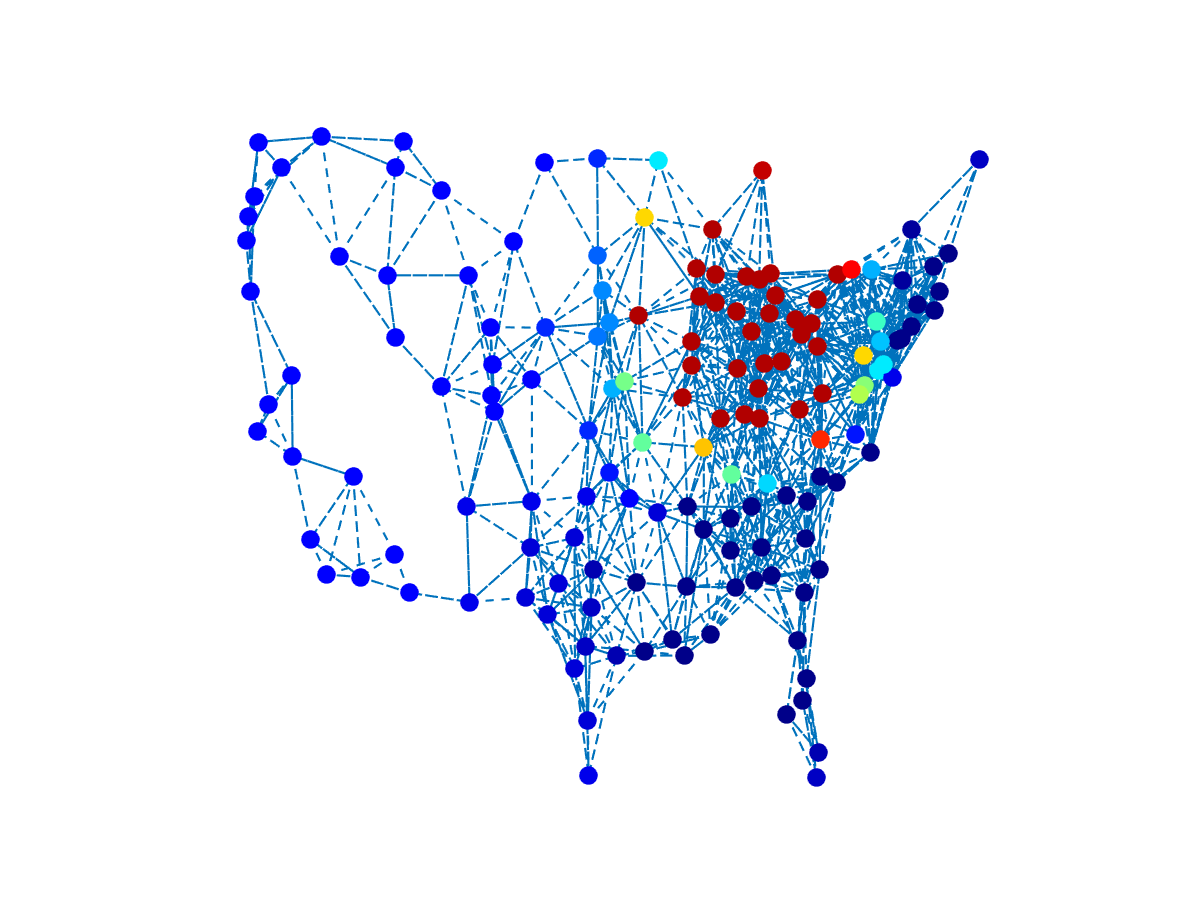}   &
\\
 {\small (g) Samples from a qualified} & {\small (h) Recovery from a qualified} & 
  \\
  {\small  sampling operator.} & {\small sampling operator.} &
\end{tabular}
  \end{center}
   \caption{\label{fig:geo_graph_signals} Graph signals on the sensor network. Colors indicate values of signal coefficients. Red represents high values, blue represents low values. Big nodes in (e) and (g) represent the sampled nodes in each case.}
\end{figure*}

\subsection{Random Sampling}
Previously, we showed that we need to design a  qualified sampling operator  to achieve perfect recovery. We now show that, for some specific graphs, random sampling leads to perfect recovery with high probabilities.

\subsubsection{Frames with Maximal Robustness to
Erasures}
A frame $\{ \f_0, \f_2, \cdots, \f_{N-1} \}$ is a generating system for $\C^K$ with $N \geq K$, when there exist two constants $a >0$ and $b < \infty$, such that for all $x \in \C^N$,
\begin{displaymath}
a \left\| \x \right\|^2 \leq   \sum_k | \f_k^T \x |^2   \leq   b \left\| \x \right\|^2.
\end{displaymath}

In finite dimensions, we represent the frame as an $N \times K$ matrix with rows $\f_k^T$. The frame is \emph{maximally robust to erasures} when every $K \times K$ submatrix (obtained by deleting $N-K$ rows) is invertible~\cite{PueschelK:05}. In~\cite{PueschelK:05}, the authors show that a polynomial transform matrix is one example of a frame maximally robust to erasures; in~\cite{SandryhailaCMKP:10}, the authors show that many lapped orthogonal transforms and lapped tight frame transforms are also maximally robust to erasures. It is clear that if the inverse graph Fourier transform matrix $\Vm$ as in~\eqref{eq:eigendecomposition} is maximally robust to erasures, any sampling operator that samples at least $K$  signal coefficients guarantees perfect recovery; in other words, when a graph Fourier transform matrix happens to be a polynomial transform matrix, sampling any $K$  signal coefficients leads to perfect recovery. 

For example, a circulant graph is a graph whose adjacency matrix is circulant~\cite{EkambaramFAR:13}. The circulant graph shift, $\Cm$, can be represented as a polynomial of the cyclic permutation matrix, $\Adj$. The graph Fourier transform of the cyclic permutation matrix is the discrete Fourier transform, which is again a polynomial transform matrix. As described above, we have
\begin{eqnarray}
\Cm & = & \sum_{i =0}^{L-1} h_i \Adj^i \ = \ \sum_{i =0}^{L-1} h_i  (\DFT^* \Lambda \DFT)^i 
\nonumber \\ \nonumber
& = & \DFT^* \left( \sum_{i =0}^{L-1} h_i  \Lambda^i  \right) \DFT,
\end{eqnarray}
where $L$ is the order of the polynomial, and $h_i$ is the coefficient corresponding to the $i$th order. Since the graph Fourier transform matrix of a circulant graph is the discrete Fourier transform matrix, we can perfectly recover a circulant-graph signal with bandwidth $K$ by sampling any $M \geq K$ signal coefficients as shown in Theorem~\ref{thm:DTPR}. In other words, perfect recovery is guaranteed when we randomly sample a sufficient  number of signal coefficients.

\subsubsection{Erd\H{o}s-R\'enyi Graph}
An Erd\H{o}s-R\'enyi graph is constructed by connecting nodes randomly, where each edge is included in the graph with probability $p$ independent of any other edge~\cite{Jackson:08,Newman:10}. We aim to show that by sampling $K$ signal cofficients randomly, the singular values of the corresponding $\Psi \Vm_{(K)}$  are bounded.

\begin{myLem}
  \label{lem:ERB}
 Let a graph shift $\Adj \in \R^{N \times N}$ represent an Erd\H{o}s-R\'enyi graph, where each pair of vertices is connected randomly and independently with probability $p = g(N)\log(N)/N$, and $g(\cdot)$ is some positive function. Let $\Vm$ be the eigenvector matrix of $\Adj$ with $\Vm^T\Vm = N\Id$, and let the number of sampled coefficients satisfy
\begin{displaymath}
 M \geq K  \frac{\log^{2.2} g(N) \log(N)}{p}  \max(C_1 \log K, C_2\log \frac{3}{\delta}),
\end{displaymath}
 for some positive constants $C_1, C_2$. Then, 
 \begin{equation}
 \label{eq:erbound}
 P \left( \left\|\frac{1}{M} (\Psi \Vm_{(K)})^T(\Psi \Vm_{(K)}) - \Id \right\|_2 \leq \frac{1}{2}  \right) \leq 1-\delta
 \end{equation}
is satisfied for all sampling operators $\Psi$ that sample $M$ signal coefficients.
\end{myLem}

\begin{proof}
For an Erd\H{o}s-R\'enyi graph, the eigenvector matrix $\Vm$ satisfies 
$$ 
\max_{i,j} |\Vm_{i,j}| = O \left(\sqrt{\log^{2.2} g(N) \log N /p} \right)
$$ for $p = g(N)\log(N)/N$~\cite{TranVW:13}. By substituting $\Vm$ into Theorem 1.2 in~\cite{CandesR:07}, we obtain~\eqref{eq:erbound}.
\end{proof}

\begin{myThm}
  \label{thm:ERframe}
  Let $\Adj, \Vm, \Psi$ be defined as in Lemma~\ref{lem:ERB}. With probability $(1- \delta)$, $\Psi \Vm_{(K)}$ is a frame in $\C^K$ with lower bound $M/2$ and upper bound $3M/2$.
\end{myThm}

\begin{proof}
Using Lemma~\ref{lem:ERB}, with  probability $(1- \delta)$, we have
   \begin{eqnarray}
   \nonumber
& \left\| \frac{1}{M} (\Psi \Vm_{(K)})^T(\Psi \Vm_{(K)}) - \Id \right\|_2  & \leq \frac{1}{2}.
   \end{eqnarray}
We thus obtain for all $\x \in \C^K$,
   \begin{eqnarray}
-\frac{1}{2} \x^T \x  \leq &  \x^T  \left( \frac{1}{M}  (\Psi \Vm_{(K)})^T(\Psi \Vm_{(K)}) - \Id \right) \x & \leq \frac{1}{2} \x^T \x,
\nonumber \\ \nonumber
\frac{M}{2} \x^T \x  \leq &  \x^T (\Psi \Vm_{(K)})^T(\Psi \Vm_{(K)})  \x & \leq \frac{3M}{2} \x^T \x.
   \end{eqnarray}
\end{proof}
 From Theorem~\ref{thm:ERframe}, we see that the singular values of  $\Psi \Vm_{(K)}$  are bounded with high probability. This shows that $\Psi \Vm_{(K)}$ has full rank with high probability; in other words, with high probability,  perfect recovery is achieved for Erd\H{o}s-R\'enyi  graph signals  when we randomly sample a sufficient number of signal coefficients.
 
\subsubsection{Simulations}
We verify Theorem~\ref{thm:ERframe} by checking the probability of satisfying the full-rank assumption by random sampling on Erd\H{o}s-R\'enyi graphs.  Once the full-rank assumption is satisfied, we can find a qualified sampling operator to achieve perfect recovery, thus, we call this probability the \emph{success rate of perfect recovery}.

We check the success rate of perfect recovery with various sizes and connection probabilities. We vary the size of an Erd\H{o}s-R\'enyi graph from 50 to 500 and the connection probabilities with an interval of 0.01 from 0 to 0.5. For each given size and connection probability, we generate 100 graphs randomly. Suppose that for each graph, the corresponding graph signal is of fixed bandwidth $K = 10$. Given a graph shift, we randomly sample 10 rows from the first 10 columns of the graph Fourier transform matrix and check whether the $10 \times 10$ matrix is of full rank. Based on Theorem~\ref{thm:GPR}, if the $10 \times 10$ matrix is of full rank, the perfect recovery is guaranteed. For each given graph shift, we run the random sampling for 100 graphs, and count the number of successes to obtain the success rate.

Figure~\ref{fig:success} shows success rates for sizes 50 and 500 averaged over 100 random tests. When we fix the graph size, in Erd\H{o}s-R\'enyi graphs, the success rate increases as the connection probability increases, that is, more connections lead to higher probability of getting a qualified sampling operator. When we compare the different sizes of the same type of graph, the success rate increases as the size increases, that is, larger graph sizes lead to higher probabilities of getting a qualified sampling operator. Overall, with a sufficient number of connections, the success rates are close to $100\%$. The simulation results suggest that the full-rank assumption is easier to satisfy when there exist more connections on graphs. The intuition is that in a larger graph with a higher connection probability, the difference between nodes is smaller, that is, each node has a similar connectivity property and there is no preference to sampling one rather than the other.

\begin{figure}[!htbp]
  \begin{center}
          \includegraphics[width=0.8\columnwidth]{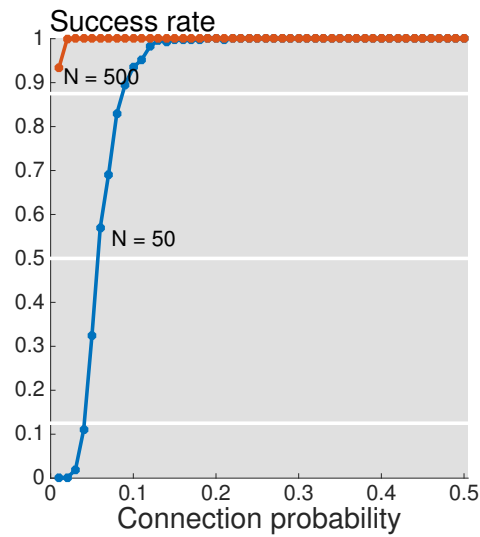} 
  \end{center}
  \caption{\label{fig:success}  Success rates for Erd\H{o}s-R\'enyi graphs. The blue curve represents an  Erd\H{o}s-R\'enyi graph of size 50 and the red curve represents an  Erd\H{o}s-R\'enyi graph of size 500.}
\end{figure}

\section{Relations \& Extensions}
\label{sec:discuss}
We now discuss three topics: relation to the sampling theory for finite discrete-time signals, relation to compressed sensing, and how to handle a full-band graph signal.

\subsection{Relation to Sampling Theory for Finite Discrete-Time Signals}
\label{sec:FDTS}
We call the graph that supports a finite discrete-time signal a~\emph{finite discrete-time graph}, which specifies the time-ordering from the past to future. The finite discrete-time graph can be represented by the cyclic permutation matrix~\cite{VetterliKG:12,SandryhailaM:131},
\begin{eqnarray}
\label{eq:CPM}
 \Adj & = &  \begin{bmatrix}
  0 & 0 &  \cdots & 1 \\
  1 & 0 &  \cdots & 0 \\
  \vdots &  \ddots & \ddots  & 0 \\
  0 & \cdots & 1  & 0
 \end{bmatrix} 
\ = \  \Vm \Lambda \Vm^{-1},
\end{eqnarray}
where the eigenvector matrix 
\begin{equation}
\label{eq:eigenvector}
\Vm 	= \begin{bmatrix} \vv_0 & \vv_1 & \cdots  & \vv_{N-1} \end{bmatrix} = \begin{bmatrix}  \frac{1}{\sqrt{N}} (W^{jk})^* \end{bmatrix}_{j, k = 0, \cdots N-1}
\end{equation}
is the Hermitian transpose of the $N$-point discrete Fourier transform matrix, $\Vm = \DFT^*$, $\Vm^{-1}$ is the $N$-point discrete Fourier transform matrix, $\Vm^{-1} = \DFT$, and the eigenvalue matrix is
\begin{equation}
\label{eq:eigenvalue}
 \Lambda  = {\rm diag} \begin{bmatrix}
		 W^{0} &   W^{1} & \cdots &  W^{N-1}
 	\end{bmatrix},
\end{equation}
where $W = e^{-2\pi j / N}$. We see that Definitions~\ref{df:GBL},~\ref{df:GBLS} and Theorem~\ref{thm:GPR} are immediately applicable to finite
discrete-time signals and are consistent with sampling of such signals~\cite{VetterliKG:12}.
\begin{defn}
  \label{df:DBL}
  A discrete-time signal is called~\emph{bandlimited} when there exists $K \in
  \{0, 1, \cdots, N-1\}$ such that its discrete Fourier transform
  $\widehat{\x}$ satisfies
  \begin{displaymath}
  \widehat{x}_k  =  0 \quad {\rm for~all~}  \quad k \geq K.
\end{displaymath}
The smallest such $K$ is called the~\emph{bandwidth} of $\x$. A discrete-time
signal that is not bandlimited is called a \emph{full-band discrete-time signal}.
\end{defn}
\begin{defn}
  \label{df:DBLS}
  The set of discrete-time signals in $\C^N$ with bandwidth of at most $K$ is
  a closed subspace denoted $\BL_K(\DFT)$, with $\DFT$ the discrete Fourier transform matrix.
\end{defn}
With this definition of the discrete Fourier
transform matrix, the highest frequency is in the middle of the
spectrum (although this is just a matter of ordering). From
Definitions~\ref{df:DBL} and~\ref{df:DBLS}, we can  permute the
rows in the discrete Fourier transform matrix to choose any frequency
band. Since the discrete Fourier transform matrix is a Vandermonde matrix, any $K$ rows of $\DFT^*_{(K)}$ are independent~\cite{HornJ:85,VetterliKG:12}; in other words, rank($\Psi \DFT^*_{(K)}) =K$ always holds when $M \geq K$. We apply now Theorem~\ref{thm:GPR} to obtain the following result.

\begin{myCorollary}
  \label{thm:DTPR}
  Let $\Psi$ satisfy that the sampling number is no smaller than the bandwidth, $M \geq K$.
For all  $\x \in \BL_K(\DFT)$, perfect recovery,  $\x = \Phi \Psi \x$, is achieved by choosing
\begin{displaymath}
   \Phi = \DFT^*_{(K)} \Um,
\end{displaymath}   
    with $\Um \Psi
  \DFT^*_{(K)}$  a $K \times K$ identity matrix, and  $\DFT^*_{(K)}$ denotes the first $K$ columns of $\DFT^*$.
\end{myCorollary}
From Corollary~\ref{thm:DTPR}, we can perfectly recover a discrete-time signal when it is bandlimited. 

Similarly to Theorem~\ref{thm:sg}, we can show that a new graph shift can be constructed from the finite discrete-time graph. Multiple sampling mechanisms can be used to sample a new graph shift; an intuitive one is as follows: let $\x \in \C^N$ be a finite discrete-time signal, where  $N$ is even. Reorder the frequencies in~\eqref{eq:eigenvalue}, by putting the frequencies with even indices first,
\begin{displaymath}
 \widetilde{\Lambda}  = {\rm diag} \begin{bmatrix}
		\lambda_0 &  \lambda_2 & \cdots & \lambda_{N-2} & \lambda_1 &  \lambda_3 & \cdots & \lambda_{N-1}
 	\end{bmatrix}.
\end{displaymath}
Similarly, reorder the columns of $\Vm$ in~\eqref{eq:eigenvector} by putting the columns with even indices first
  \begin{displaymath}
\widetilde{\Vm} 	 = \begin{bmatrix}  \vv_0 & \vv_2 & \cdots & \vv_{N-2}  &   \vv_1 & \vv_3 & \cdots & \vv_{N-1} \end{bmatrix}.
\end{displaymath}
One can check that $\widetilde{\Vm} \widetilde{\Lambda} \widetilde{\Vm} ^{-1} $ is still the same cyclic permutation matrix. Suppose we want to preserve the first $N/2$ frequency components in  $\widetilde{\Lambda}$; the sampled frequencies are then
\begin{displaymath}
 \widetilde{\Lambda}_{(N/2)}  = {\rm diag} \begin{bmatrix}
		\lambda_0 &  \lambda_2 & \cdots & \lambda_{N-2}
 	\end{bmatrix}.
\end{displaymath}
Let a sampling operator $\Psi$ choose the first $N/2$ rows in $\widetilde{\Vm}_{(N/2)}$,  
\begin{displaymath}
	\Psi \widetilde{\Vm}_{(N/2)}	 = 
\begin{bmatrix}
 \frac{1}{\sqrt{N}} (W^{2jk})^* \end{bmatrix}_{j, k = 0, \cdots N/2-1},
\end{displaymath}
which is the Hermitian transpose of the discrete Fourier transform of size $N/2$ and satisfies rank$(\Psi \widetilde{\Vm}_{(N/2)})=N/2$ in Theorem~\ref{thm:sg}.  The sampled graph Fourier transform matrix $\Um = (\Psi \widetilde{\Vm}_{(N/2)}))^{-1}$ is the discrete Fourier transform of size $N/2$. The sampled graph shift is then constructed as
\begin{displaymath}
\Adj_\M = \Um^{-1} \widetilde{\Lambda}_{(N/2)} \Um,
\end{displaymath}
which is exactly the $N/2 \times N/2$ cyclic permutation matrix.  Hence, we have shown that  by choosing an appropriate sampling mechanism, a smaller finite discrete-time graph is obtained from a larger finite discrete-time graph by using Theorem~\ref{thm:sg}. We note that  using a different ordering or  sampling operator would result in a graph shift that can be different and non-intuitive. This is simply a matter of choosing different frequency components.

\subsection{Relation to Compressed Sensing}
Compressed sensing is a sampling framework to recover sparse signals with a few measurements~\cite{Donoho:06}. The theory asserts that a few samples guarantee the recovery of the original signals when the signals and the sampling approaches are well-defined in some theoretical aspects. To be more specific, given the sampling operator $\Psi \in \mathbb{R}^{M \times N}, M \ll N $, and the sampled signal $\x_\M \ = \ \Psi \x$, a sparse signal $\x \in \mathbb{R}^N$ is recovered by solving  
\begin{eqnarray}
\min_\x {~ \left\| \x \right\|_{0}}, {\rm~~subject~to~} \x_\M \ = \ \Psi \x. \end{eqnarray}
Since the $l_0$ norm is not convex, the optimization is a non-deterministic polynomial-time hard problem. To obtain a computationally efficient algorithm,  the $l_1$-norm based algorithm, known as basis pursuit or basis pursuit with denoising, recovers the sparse signal with small approximation error~\cite{ChenDS:01}.

In the standard compressed sensing theory, signals have to be sparse or approximately sparse to gurantee accurate recovery properties.  In~\cite{CandesENR:10}, the authors proposed a general way to perform compressed sensing with non-sparse signals using dictionaries. Specifically, a general signal $\x \in \mathbb{R}^N$ is recovered by 
\begin{eqnarray}
\label{eq:gcs}
 \min_\x {~\left\|  \D \x  \right|_{0}}, {\rm~~subject~to~} \x_\M = \Psi \x,
\end{eqnarray}
where $\D$ is a dictionary designed to make $\D \x$ sparse. When specifying $\x$ to be a graph signal, and $\D$ to be the appropriate graph Fourier transform of the graph on which the signal resides, $\D\x$ represents the frequency content of $\x$, which is sparse when $\x$ is of limited bandwidth. Equation~\eqref{eq:gcs} recovers a bandlimited graph signal from a few sampled signal coefficients via an optimization approach. The proposed sampling theory deals with the cases where the frequencies corresponding to non-zero elements are known, and can be reordered to form a bandlimited graph signal. Compressed sensing deals with the cases where the frequencies corresponding to non-zero elements are unknown, which is a more general and harder problem. If we have access to the position of the non-zero elements, the proposed sampling theory uses the smallest number of samples to achieve perfect recovery.

\subsection{Relation to Signal Recovery on Graphs}
Signal recovery on graphs attempts to recover graph signals that are assumed to be smooth with respect to underlying graphs, from noisy, missing, or corrupted samples~\cite{ChenSMK:14}. Previous works studied signal recovery on graphs from different perspectives. For example, in~\cite{BelkinN:04}, the authors considered that a graph is a discrete representation of a manifold, and aimed at recovering graph signals by regularizing the smoothness functional; in~\cite{GradyS:03}, the authors aimed at recovering graph signals by regularizing the combinatorial Dirichlet; in~\cite{ChenSMK:14}, the authors aimed at recovering graph signals by regularizing the graph total variation;  in~\cite{TremblayBF:14}, the authors aimed at recovering graph signals by finding the empirical modes;
in~\cite{ThanouSF:14}, the authors aimed at recovering graph signals by training a graph-based dictionary. These works focused on minimizing the empirical recovery error, and dealt with general graph signal models. It is thus hard to show when and how the missing signal coefficients can be exactly recovered. 

Similarly to signal recovery on graphs, sampling theory on graphs also attempts to recover graph signals from incomplete samples. The main difference is that sampling theory on graphs focuses on a subset of smooth graph signals, that is, bandlimited graph signals, and theoretically analyzes when and how the missing signal coefficients can be exactly recovered. The authors in~\cite{Pesenson:08,Pesenson:09, AnisGO:14} considered a similar problem to ours, that is, recovery of bandlimited graph signals by sampling a few signal coefficients in the vertex domain. The main differences are as follows: (1) we focus on the graph adjacency matrix, not the graph Laplacian matrix; (2) when defining the bandwidth, we focus on the number of frequencies, not the values of frequencies. Some recent extensions of sampling theory on graphs include~\cite{FuharP:13, ChenVSK:15, GaddeO:15}.

\subsection{Graph Downsampling \& Graph Filter Banks}
In classical signal processing, sampling refers to sampling a continuous function and downsampling refers to sampling a sequence. Both concepts use fewer samples to represent the overall shape of the original signal. Since a graph signal is discrete in nature, sampling and downsampling are the same. Previous works implemented graph downsampling  via graph coloring~\cite{NarangO:12} or minimum spanning tree~\cite{NguyenD:15}.

The proposed sampling theory provides a family of qualified sampling operators~\eqref{eq:Assumption} with an optimal sampling operator as in~\eqref{eq:optimalset}. To downsample a graph by 2,  one can set the bandwidth to a half of the number of nodes, that is, $K = N/2$, and use~\eqref{eq:optimalset} to obtain an optimal sampling operator. An example for the finite discrete-time signals was shown in Section~\ref{sec:FDTS}.

As shown in Theorem~\ref{thm:GPR}, perfect recovery is achieved when graph signals are bandlimited. To handle full-band graph signals, we propose an approach based on graph filter banks, where each channel does not need to recover perfectly but in conjunction they do.

\begin{figure*}[t]
  \begin{center}
     \includegraphics[width= 1.3\columnwidth]{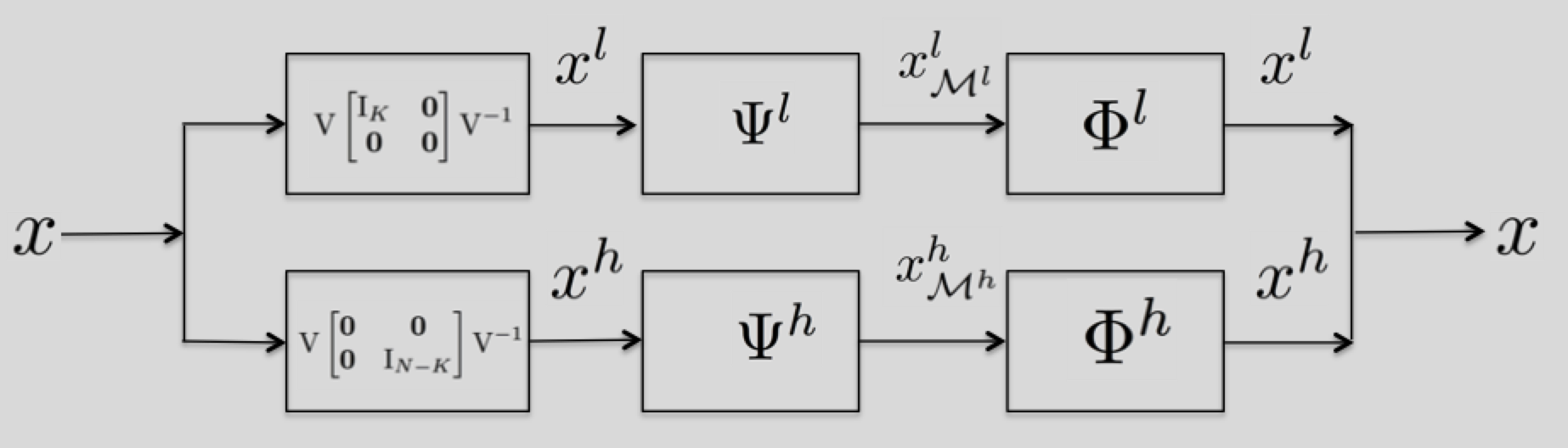}
  \end{center}
  \caption{\label{fig:filterbanks}  Graph filter bank that splits the graph signal into two bandlimited graph signals. In each channel, we perform sampling and interpolation, following Theorem~\ref{thm:GPR}. Finally, we add the results from both channels to obtain the original full-band graph signal.}
\end{figure*}

Let $\x$ be a full-band graph signal, which, without loss of generality, we can express without loss of generality as the addition of two bandlimited signals supported on the same graph, that is, $\x = \x^l + \x^h$, where
\begin{displaymath}
\x^l  = \Pj^l \x,~~
\x^h  = \Pj^h  \x,
\end{displaymath}
and 
\begin{displaymath}
\Pj^l \ = \ \Vm \begin{bmatrix}
\Id_K & \bold{0} \\ \bold{0}   & \bold{0}
\end{bmatrix} \Vm^{-1},
~~\Pj^h \ = \ = \Vm \begin{bmatrix}
\bold{0}  & \bold{0} \\ \bold{0}   & \Id_{N-K} 
\end{bmatrix} \Vm^{-1}.
\end{displaymath}
We see that $\x^l$ contains the first $K$ frequencies, $\x^h$ contains the other $N-K$ frequencies, and each is bandlimited. We do sampling and interpolation for $\x^l$ and $\x^h$ in two channels, respectively. Take the first channel as an example. Following Theorems~\ref{thm:GPR} and~\ref{thm:sg}, we use a qualified sampling operator $\Psi^l$ to sample $\x^l$, and obtain the sampled signal coefficients as $\x^l_{\M^l} = \Psi^l \x^l$, with the corresponding graph as $\Adj_{\M^l}$. We can recover $\x^l$ by using interpolation operator $\Phi^l$ as $\x^l = \Phi^l \x^l_{\M^l}$. Finally, we add the results from both channels to obtain the original full-band graph signal (also illustrated in Figure~\ref{fig:filterbanks}). The main benefit of working with a graph filter bank is that, instead of dealing with a long graph signal with a large graph, we are allowed to focus on the frequency bands of interests and deal with a shorter graph signal with a small graph in each channel.

We do not restrict the samples from two bands,  $\x^l_{\M^l}$ and $\x^h_{\M^h}$ the same size because we can adaptively design the sampling and interpolation operators based on the their own sizes. This is similar to the filter banks in the classical literature where the spectrum is not evenly partitioned between the channels~\cite{Vetterli:87}. We see that the above idea can easily be generalized to multiple channels by splitting the original graphs signal into multiple bandlimited graph signals; instead of dealing with a huge graph, we work with multiple small graphs, which makes computation easier.

\mypar{Simulations}
We now show an example where we analyze graph signals by using the proposed graph filter banks. Similarly to Section~\ref{sec:sensor_optimal}, we consider that the weather stations across the U.S. form a graph and temperature values measured at each weather station in one day form a graph signal. Suppose that a high-frequency component represents some pattern of weather change; we want to detect this pattern given the temperature values. We can decompose a graph signal of temperature values into a low-frequency channel (largest 15 frequencies) and a high-frequency channel (smallest 5 frequencies). In each channel, we sample the bandlimited graph signal to obtain a sparse and loseless representation. Figure~\ref{fig:geo_GFB} shows a comparison between temperature values on January 1st, 2013 and May 1st, 2013. We intentionally added some high-frequency component to the temperature values on January 1st, 2013. We see that the high-frequency channel in Figure~\ref{fig:geo_GFB} (a) detects the change, while  the high-frequency channel in Figure~\ref{fig:geo_GFB} (b) does not.

Directly using the graph frequency components can also detect the high-frequency components, but the graph frequency components cannot be easily visualized. Since the graph structure and the decomposed channels are fixed, the optimal sampling operator and corresponding interpolation operator in each channel can be designed in advance, which means that we just need to look at the sampled coefficients of a fixed set of nodes to check whether a channel is activated. The graph filter bank is thus fast and visualization friendly.

\begin{figure*}[htb]
  \begin{center}
    \begin{tabular}{cc}
\includegraphics[width=1\columnwidth]{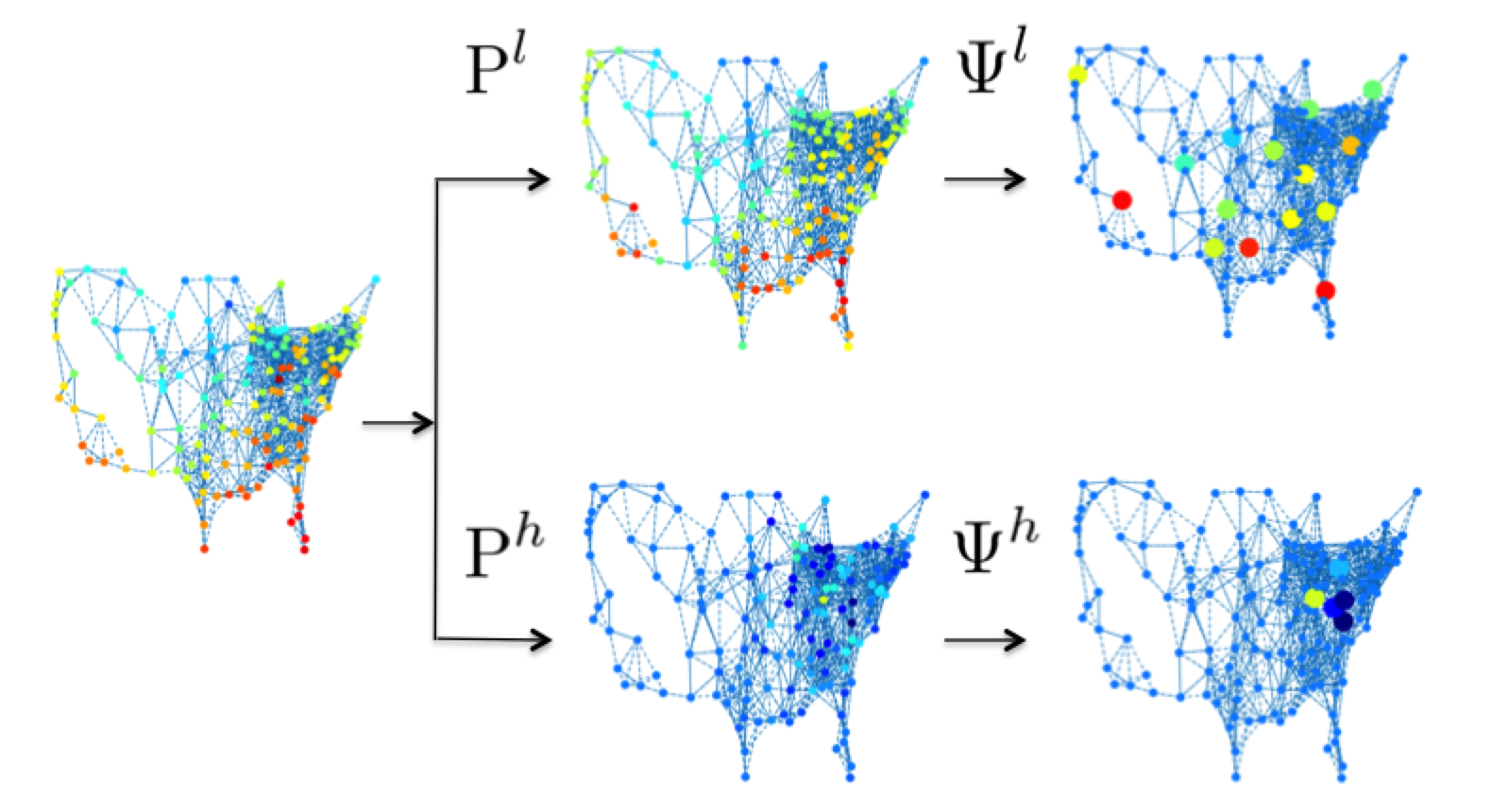} 
&
 \includegraphics[width=1\columnwidth]{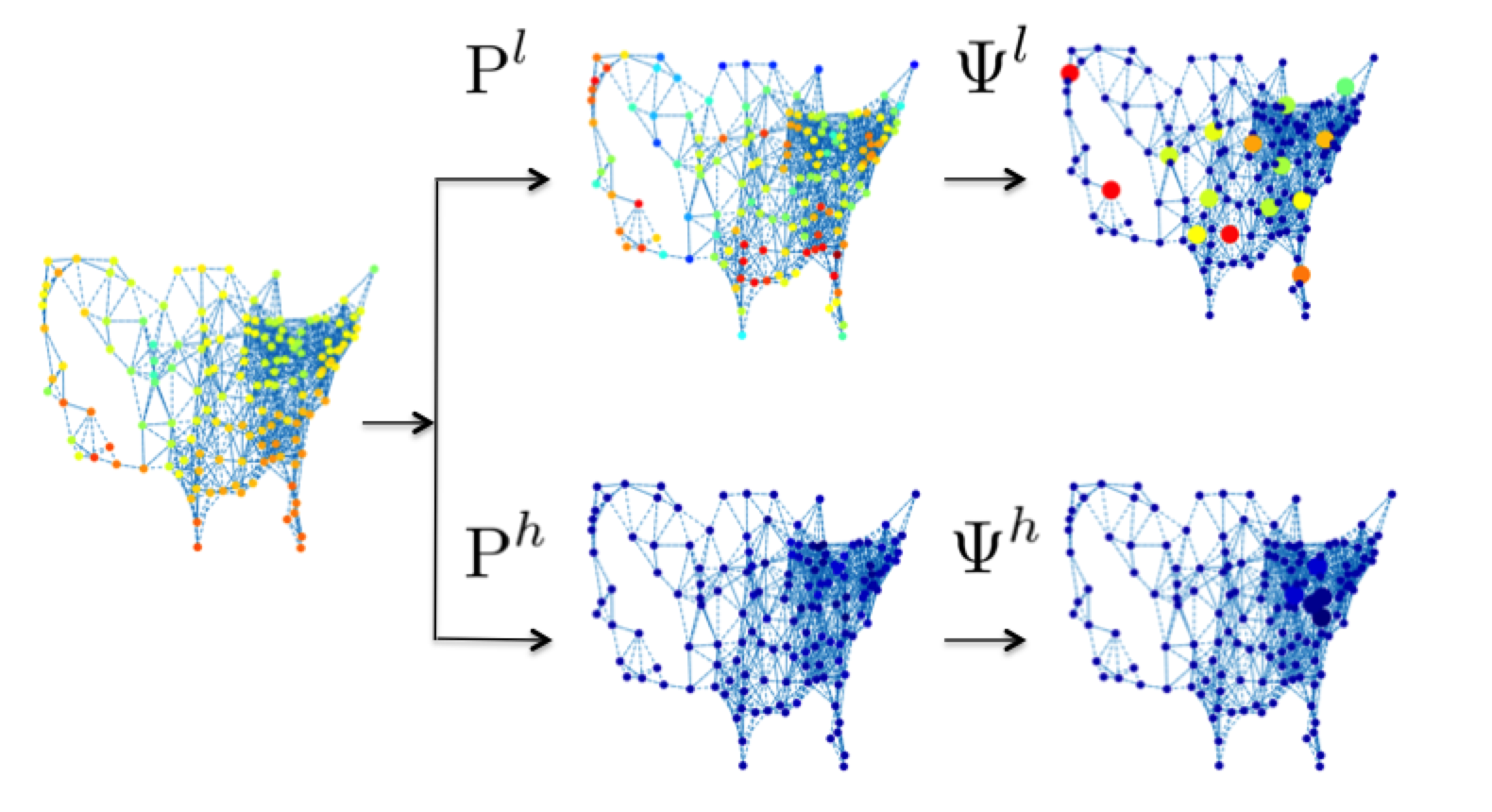} 
\\
 {\small (a) Temperature on January 1st, 2013 with high-frequency pattern.}  &  {\small (b) Temperature on May 1st, 2013.}
 \\
\end{tabular}
  \end{center}
   \caption{\label{fig:geo_GFB} Graph filter banks analysis. }
\end{figure*}




\section{Applications}
\label{sec:apps}
The proposed sampling theory on graphs can be applied to semi-supervised learning, whose goal is to classify data with a few labeled and a large number of unlabeled samples~\cite{Zhu:05}. One approach is based on graphs, where the labels are often assumed to be smooth on graphs. From a perspective of signal processing, smoothness can be expressed as lowpass nature of the signal. Recovering smooth labels on graphs is then equivalent to interpolating a low-pass graph signal. We now look at two examples, including classification of online blogs and handwritten digits.

\subsubsection{Sampling Online Blogs}
\label{sec:blogs}
We first aim to  investigate the success rate of perfect recovery by using random sampling, and then classify the labels of the online blogs. We consider a dataset of $N=1224$ online political blogs
as either conservative or liberal~\cite{AdamicG:05}. We represent
conservative labels as $+1$ and liberal ones as $-1$.  The blogs are
represented by a graph in which nodes represent blogs, and directed
graph edges correspond to hyperlink references between blogs. The graph signal here is the label assigned to the blogs, called the labeling signal.  We use the spectral decomposition in~\eqref{eq:eigendecomposition} for this online-blog graph to get the graph frequencies in a  descending order and the corresponding graph Fourier transform matrix. The labeling signal is a full-band signal, but approximately bandlimited.

To investigate the success rate of perfect recovery by using random sampling, we vary the bandwidth $K$ of the labeling signal with an interval of 1 from 1 to 20, randomly sample $K$ rows from the first $K$ columns of the graph Fourier transform matrix, and check whether the $K \times K$ matrix has full rank. For each bandwidth, we  randomly sample 10,000 times, and count the number of successes to obtain the success rate. Figure~\ref{fig:blog} (a) shows the resulting success rate. We see that the success rate decreases as we increase the bandwidth; it is above $90\%$, when the bandwidth is no greater than $20$. It means that we can achieve perfect recovery by using random sampling with a fairly high probability. As the bandwidth increases, even if we get an equal number of samples, the success rate still decreases, because when we take more samples, it is easier to get a sample correlated with the previous samples.

\begin{figure*}[!htbp]
  \begin{center}
  \begin{tabular}{cc}  
      \includegraphics[width=0.6\columnwidth]{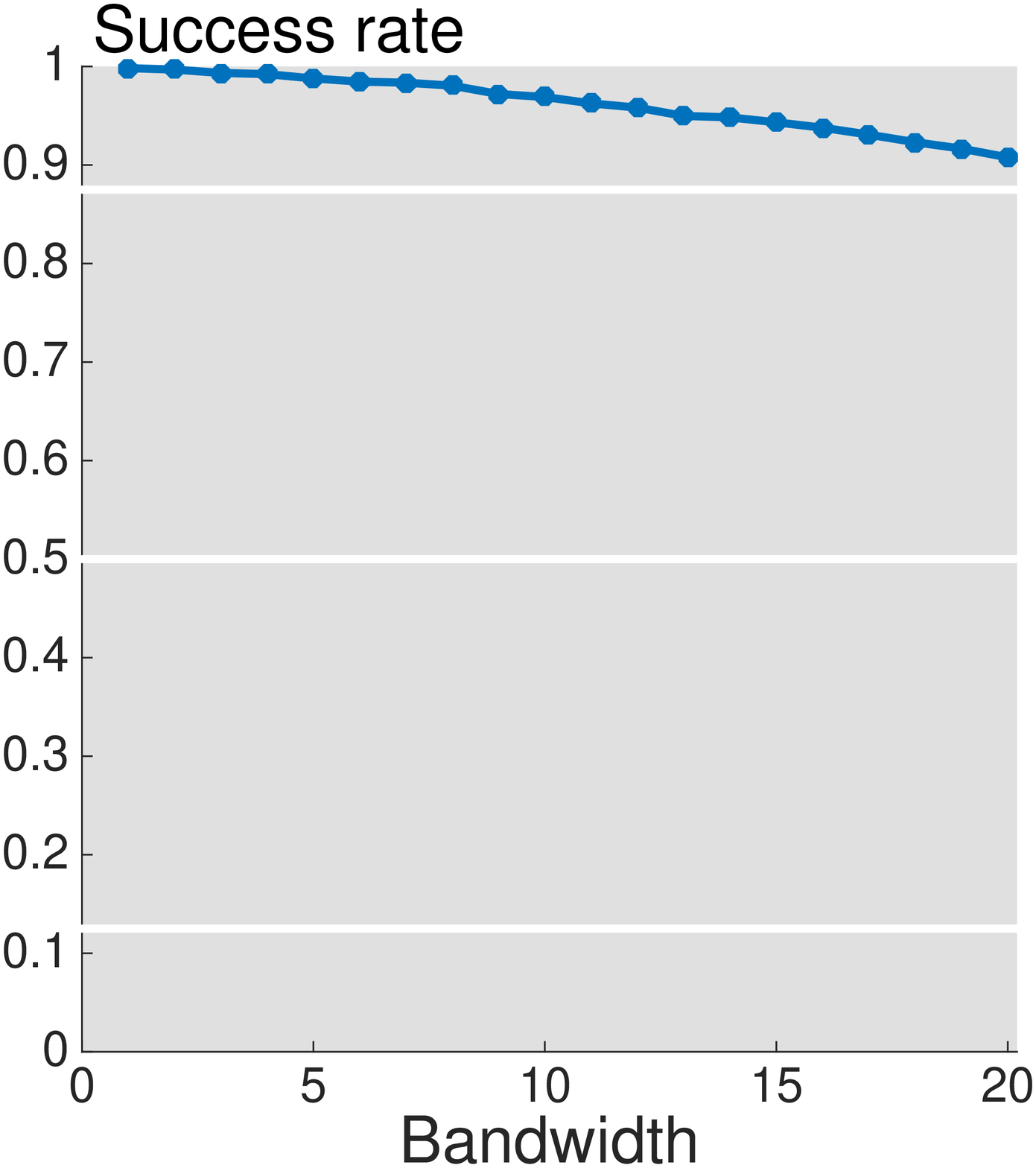} 
      &
       \includegraphics[width=0.6\columnwidth]{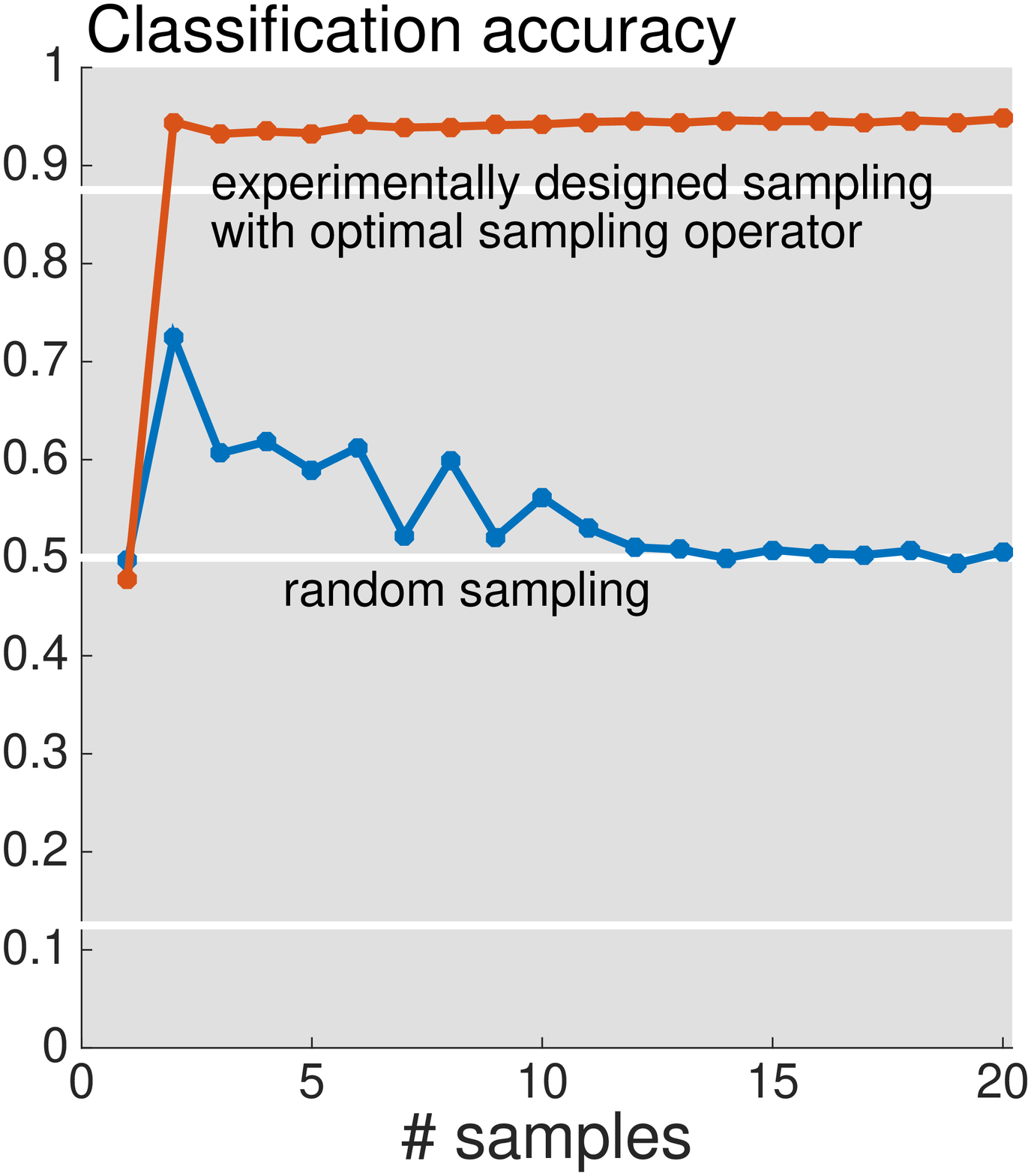}  \\
        {\small (a) Success rate as a function of bandwidth.} & {\small (b) Classification accuracy as a function of the number of samples.}
      \end{tabular}
  \end{center}
  \caption{\label{fig:blog}  Classification for online blogs.  When increasing the bandwidth, it is harder to find a qualified sampling operator. The experimentally designed sampling with the optimal sampling operator outperforms random sampling.}
\end{figure*}

Since a qualified sampling operator is independent of graph signals, we precompute the qualified sampling operator for the online-blog graph, as discussed in Section~\ref{sec:ggs}. When the labeling signal is bandlimited, we can sample $M$ labels from it by using a qualified sampling operator, and recover the labeling signal by using the corresponding interpolation operator. In other words, we can design a set of blogs to label before querying any label.  Most of the time, however, the labeling signal is not bandlimited, and it is not possible to achieve perfect recovery. Since we only care about the sign of the labels, we use only the low frequency content to approximate the labeling signal; after that, we set a threshold to assign labels. To minimize the influence from the high-frequency content, we can use the optimal sampling operator in Algorithm~\ref{alg:optimalset}.

We solve the following optimization problem to recover the low frequency content,
\begin{equation}
\label{eq:classification}
\widehat{\x}_{(K)}^{opt} \ = \ \arg\min_{\widehat{\x}_{(K)} \in \R^K} ~\left\| {\rm sgn}(\Psi \Vm_{(K)} \widehat{\x}_{(K)}) - \x_\M \right\|_2^2,
\end{equation}
where $\Psi \in \R^{M \times N}$ is a sampling operator, $\x_\M \in \R^M$ is a vector of the sampled labels whose element is either $+1$ or $-1$, and ${\rm sgn}(\cdot)$ sets all positive values to $+1$ and all negative values to $-1$.  Note that without ${\rm sgn}(*)$, the solution of~\eqref{eq:classification} is $(\Psi \Vm_{(K)})^{-1} \x_\M$ in Theorem~\ref{thm:GPR}, which perfectly recovers the labeling signal when it is bandlimited. When the labeling signal is not bandlimited, the solution of~\eqref{eq:classification} approximates the low-frequency content. The $\ell_2$ norm~\eqref{eq:classification} can be relaxed by the logit function and solved by logistic regression~\cite{Bishop:06}. The recovered labels are then $\x^{opt} ={\rm sgn} (\Vm_{(K)}\widehat{\x}_{(K)}^{opt})$. 

Figure~\ref{fig:blog} (b) compares the classification accuracies between optimal sampling and random sampling by varying the sample size  with an interval of 1 from 1 to 20. We see that the optimal sampling significantly outperforms random sampling, and random sampling does not improve with more samples, because the interpolation operator~\eqref{eq:classification} assumes that the sampling operator is qualified, which is not always true for random sampling as shown in Figure~\ref{fig:blog} (a). Note that classification accuracy for the optimal sampling is as high as $94.44\%$ by only sampling two blogs, and the classification accuracy gets slightly better as we increases the number of samples.  Compared with the previous results~\cite{ChenSLWMRBGK:14}, to achieve around $94\%$ classification accuracy,  
\begin{itemize}
\item harmonic functions on graph samples 120 blogs;
\item graph Laplacian regularization samples 120 blogs;
\item graph total variation regularization samples 10 blogs; and 
\item the proposed optimal sampling operator~\eqref{eq:optimalset} samples 2 blogs.
\end{itemize}

The improvement comes from the fact that, instead of sampling randomly as in~\cite{ChenSLWMRBGK:14}, we use the optimal sampling operator to choose samples based on the graph structure.

 \begin{figure*}
  \begin{center}
    \begin{tabular}{cc}  
          \includegraphics[width=1\columnwidth]{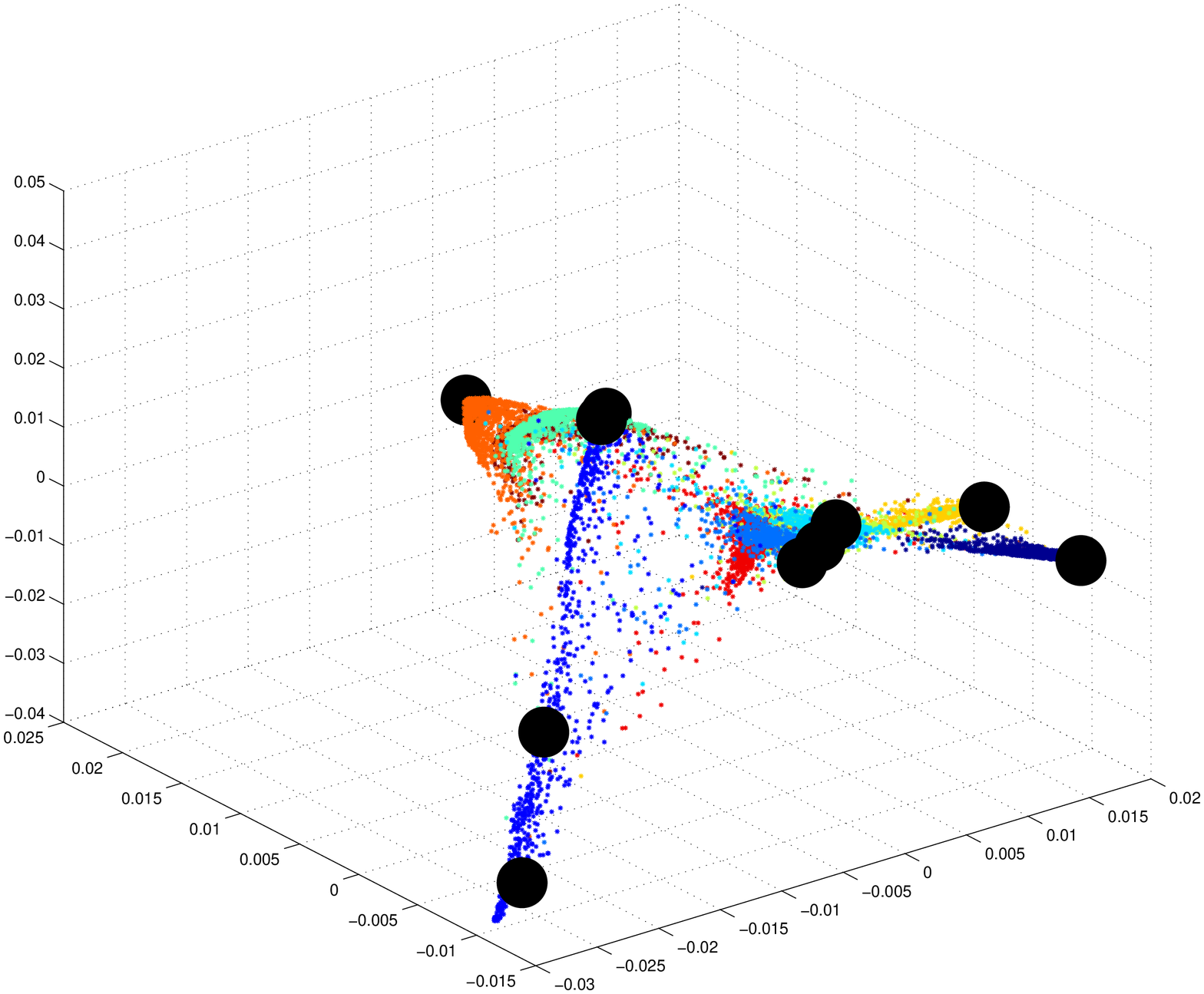} 
      & 
      \includegraphics[width=1\columnwidth]{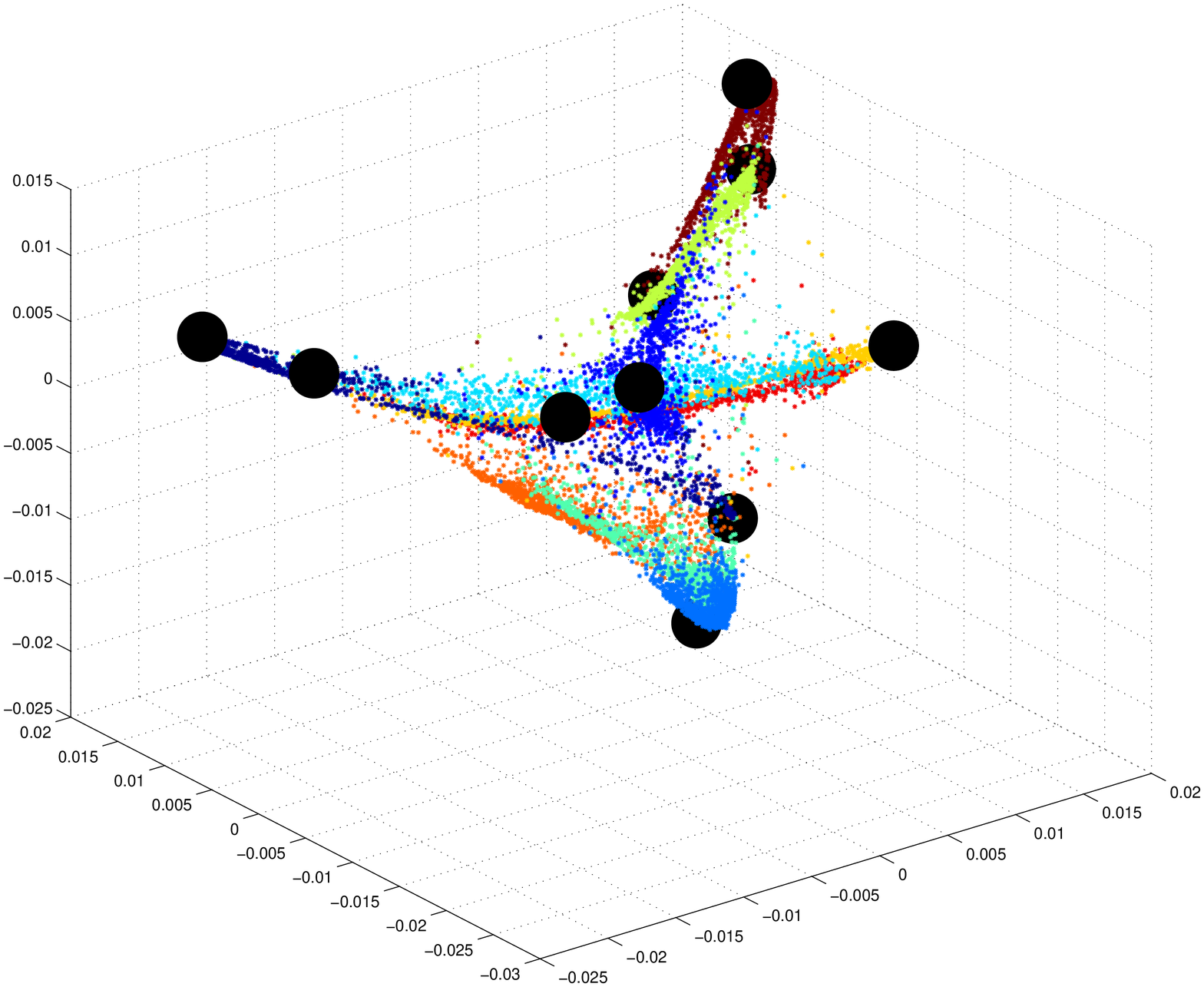}\\
      {\small (a) MNIST.} & {\small (b) USPS.}
    \end{tabular}
  \end{center}
  \caption{\label{fig:scatter}  Graph representations of the MNIST and USPS datasets. For both datasets, the nodes (digit images) with the same digit characters are shown in the same color and the big black dots indicate 10 sampled nodes by using the optimal sampling operators in Algorithm~\ref{alg:optimalset}.}
\end{figure*}

\subsubsection{Classification for Handwritten Digits}
We aim to use the proposed sampling theory to classify handwritten digits and achieve high classification accuracy with fewer samples.

We work with two handwritten digit datasets, the MNIST ~\cite{LeCun:01} and the USPS~\cite{Hull:94}. Each dataset includes ten classes (0-9 digit characters). The MNIST dataset includes 60,000 samples in total. We randomly select 1000 samples for each digit character, for a total of $N =10,000$ digit images; each image is normalized to $28 ￼\times 28 = 784$ pixels.  The USPS dataset includes 11,000 samples in total. We use all the images in the dataset;  each image is normalized to $16 ￼\times 16 = 256$ pixels. 

Since same digits produce similar images, it is intuitive to build a graph to reflect the relational dependencies among images. For each dataset, we construct a 12-nearest neighbor graph to represent  the digit images. The nodes represent digit images and each node is connected to 12 other nodes that represent the most similar digit images; the similarity is measured by the Euclidean distance. The
graph shift is constructed as $\Adj_{i,j} =
\Pj_{i,j}/\sum_i{\Pj_{i,j}}$, with
\begin{displaymath}
  \Pj_{i,j} \ = \ \exp \left( \frac{-N^2\left\|\f_i - \f_j \right\|_2} {\sum_{i,j} \left\|\f_i - \f_j \right\|_2} \right),
\end{displaymath}
with $\f_i$ a vector representing the digit image. The graph shift is asymmetric, representing a directed graph, which cannot be handled by graph Laplacian-based methods.

Similarly to Section~\ref{sec:blogs},  we aim to label all the digit images by actively querying the labels of a few images. To handle 10-class classification, we
form a ground-truth matrix $\X$ of size $N \times 10$. The element $\X_{i,j}$ is $+1$, indicating the membership of the $i$th image in the $j$th digit class, and is $-1$ otherwise. We obtain the optimal sampling operator $\Psi$ as shown in Algorithm~\ref{alg:optimalset}. The querying samples are then $\X_\M = \Psi \X \in \R^{M \times 10} $. We recover the low frequency content as
\begin{equation}
\label{eq:classification2}
\widehat{\X}_{(K)}^{opt} \ = \ \arg\min_{\widehat{\X}_{(K)} \in \R^{K \times 10}} ~\left\| {\rm sgn}(\Psi \Vm_{(K)} \widehat{\X}_{(K)}) - \X_\M \right\|_2^2.
\end{equation}
We solve~\eqref{eq:classification2} approximately by using logistic regression and then obtain the estimated label matrix $\X^{opt} = \Vm_{(K)}\widehat{\X}_{(K)}^{opt}  \in \R^{N \times 10}$, whose element $(\X^{opt})_{i,j}$ shows a confidence of labeling the $i$th image as the $j$th digit. We finally label each digit image by choosing the one with largest value in each row of $\X^{opt}$.  

The graph representations of the MNIST and USPS datasets, and the optimal sampling sets are shown in Figure~\ref{fig:scatter}. The coordinates of nodes come from the corresponding rows of the first three columns of the inverse graph Fourier transform. We see that the images with the same digit characters form clusters, and the optimal sampling operator chooses representative samples from different clusters.
\begin{figure}[!htbp]
  \begin{center}
    \begin{tabular}{cc}  
          \includegraphics[width=0.48\columnwidth]{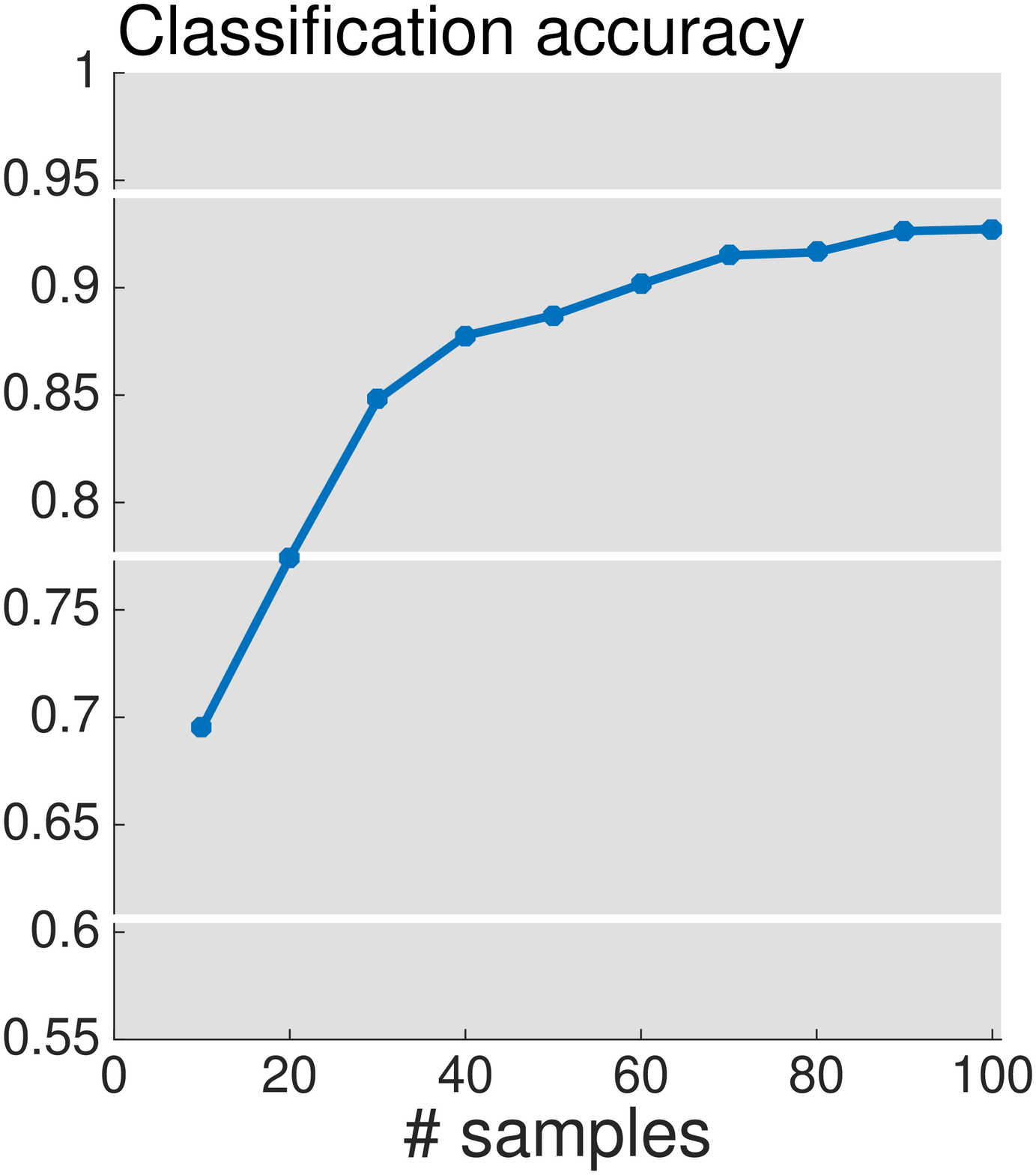} 
      & 
      \includegraphics[width=0.48\columnwidth]{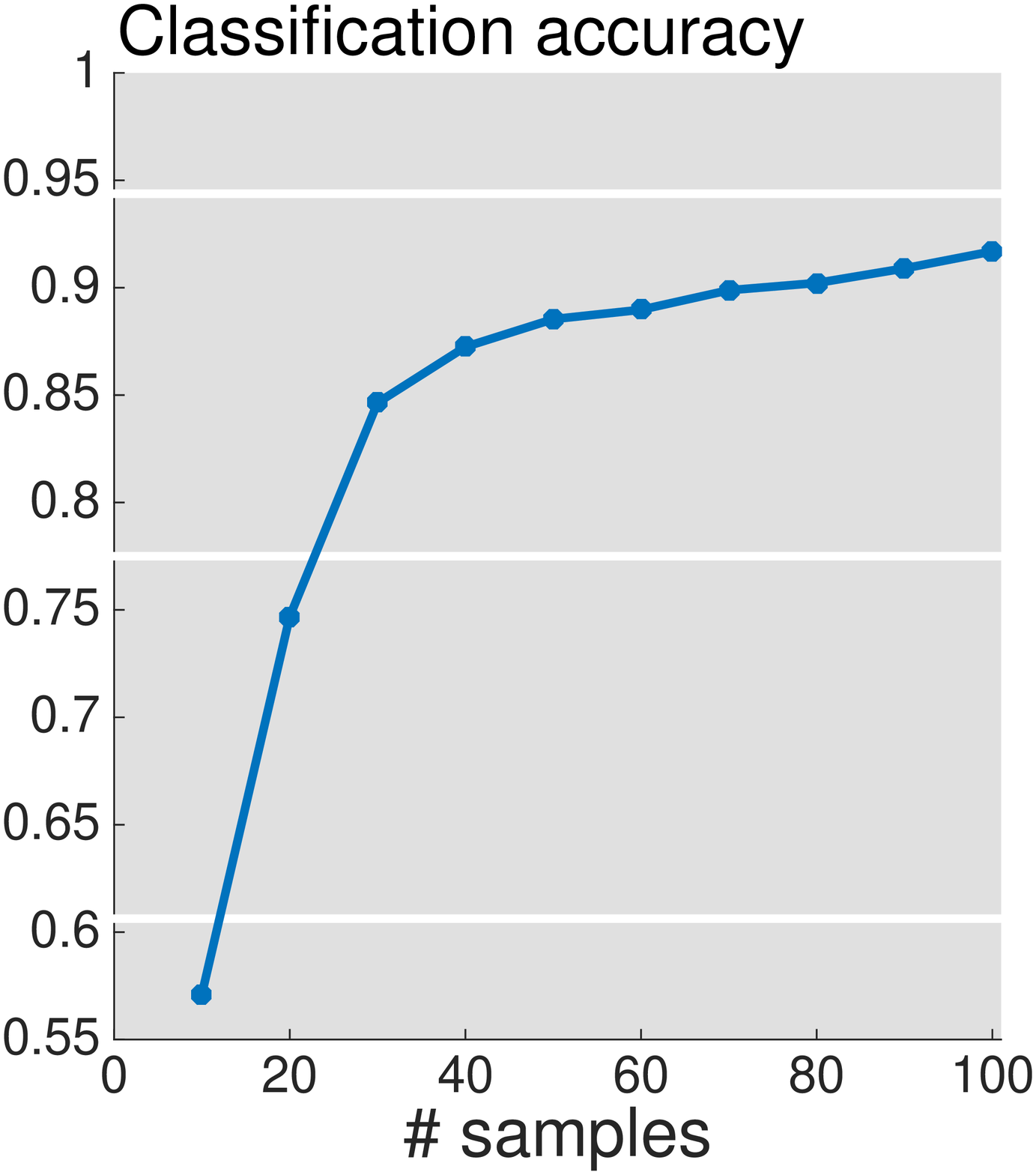}\\
      {\small (a) MNIST.} & {\small (b) USPS.}
    \end{tabular}
  \end{center}
  \caption{\label{fig:digit_classification}  Classification accuracy of the MNIST and USPS datasets as a function of the number of querying samples.}
\end{figure}

Figure~\ref{fig:digit_classification} shows the classification accuracy by varying the sample size with an interval of 10 from 10 to 100 for both datasets. For the MNIST dataset, we query $0.1\%-1\%$ images; for the USPS dataset, we query $0.09\%-0.9\%$ images. We achieve around $90\%$ classification accuracy by querying only $0.5\%$ images for both datasets. Compared with the previous results~\cite{GaddeAO:14}, in the USPS dataset, given 100 samples,
\begin{itemize}
\item local linear reconstruction is around $65\%$;
\item  normalized cut based active learning is around $70\%$;
\item graph sampling based active semi-supervised learning is around $85\%$; and
\item the proposed optimal sampling operator~\eqref{eq:optimalset} with the interpolation operator~\eqref{eq:classification2} achieves $91.69\%$.
\end{itemize}

\section{Conclusions}
\label{sec:conclusions}
We proposed a novel sampling framework for graph signals that follows the same paradigm as classical sampling theory and strongly connects to linear algebra. We showed that perfect recovery is possible when graph signals are bandlimited. The sampled signal coefficients then form a new graph signal, whose corresponding graph structure is constructed from the original graph structure, preserving the first-order difference of the original graph signal. We studied a qualified sampling operator for both random sampling and experimentally designed sampling.  We further established the connection to the sampling theory for finite discrete-time signal processing and previous works on the sampling theory on graphs, and showed how to handle full-band graphs signals by using graph filter banks. We showed applications to semi-supervised classification of online blogs and digit images, where the proposed sampling and interpolation operators perform competitively.
\bibliographystyle{IEEEbib}
\bibliography{bibl_jelena}

\end{document}